\documentclass[twocolumn,notitlepage,nofootinbib,pra,superscriptaddress,reprint,showkeys]{revtex4-2}

\usepackage{amsmath,amssymb,bbm,soul}
\usepackage{amsthm}
\usepackage{graphicx}
\usepackage[colorlinks=true,citecolor=blue,linkcolor=blue,urlcolor=blue]{hyperref}
\usepackage[nameinlink,capitalize]{cleveref}

\catcode`\|=\active \def|{
\fontencoding{T1}\selectfont\symbol{124}\fontencoding{\encodingdefault}}

\newcommand{\sys}{\text{S}}
\newcommand{\env}{\text{E}}

\newcommand{\Hil}{\mathcal{H}} 
\newcommand{\Hils}{\mathcal{H}_\sys} 
\newcommand{\Hile}{\mathcal{H}_\env} 
\newcommand{\Evl}{\mathcal{V}}
\newcommand{\K}{\mathcal{K}}
\newcommand{\Rho}{\hat{\rho}} 

\newcommand{\epsilonbm}{\hat{\epsilon}}
\newcommand{\ubs}[2]{\underset{#1}{\underbrace{#2}}}  
\newcommand{\B}{\mathcal{B}}   
\newcommand{\F}{\mathcal{F}}   
\newcommand{\E}{\mathcal{E}}   
\DeclareMathOperator{\trace}{Tr}
\newcommand{\Tr}[1]{\trace\left[#1\right]}
\newcommand{\partTr}[2]{\trace_{#1}\left[#2\right]}

\newcommand{\ketbra}[2]{|#1\rangle\!\langle#2|}
\newcommand{\und}{\text{and}}   
\newcommand{\spancvx}[1]{\mathrm{conv}\left[#1\right]}

\newcounter{theorems}
\newtheorem{thm}[theorems]{Theorem}

\newtheorem{defin}{Definition}
\newtheorem{lem}[theorems]{Lemma}

\newtheorem{ex.}{Example}[theorems]

\newtheorem*{cor*}{Corollary}
\newtheorem*{thm*}{Theorem}
\newtheorem*{prop*}{Proposition}
\newtheorem*{lem*}{Lemma}
\newtheorem*{rem*}{Remark}

\newtheorem{ass}{Condition}
\crefname{ass}{Cnd.}{Cnds.}
\newtheorem*{ass*}{Condition}
\crefname{fig}{Fig.}{Figs.}

\newcommand{\id}{\mathbbm{1}} 

\newcommand{\ket}[1]{|#1\rangle} 
\newcommand{\bra}[1]{\langle#1|} 

\usepackage{xcolor}
\usepackage[normalem]{ulem}
\newcommand{\stkout}[1]{\ifmmode\text{\sout{\ensuremath{#1}}}\else\sout{#1}\fi}

\definecolor{textblue}{RGB}{20,120,165}
\definecolor{textterra}{RGB}{144,73,26}

\begin{document}

\title{Classical Invasive Description of Informationally-Complete Quantum Processes}
	
	\author{Moritz F. Richter}
	\email{moritz.ferdinand.richter@physik.uni-freiburg.de} 
	\affiliation{Institute  of  Theoretical  Physics,  TUD Dresden University of Technology, D-01062 Dresden, Germany}
	\affiliation{Institute of Physics, University of Freiburg, Hermann-Herder-Str. 3, D-79104 Freiburg, Germany}
	
	\author{Andrea Smirne}
	\email{andrea.smirne@unimi.it}
	\affiliation{Dipartimento di Fisica “Aldo Pontremoli”, Universit\`a degli Studi di Milano, via Celoria 16, 20133 Milan, Italy}
	\affiliation{Istituto Nazionale di Fisica Nucleare, Sezione di Milano, via Celoria 16, 20133 Milan, Italy}

	\author{Walter T. Strunz}
	\email{walter.strunz@tu-dresden.de}
	\affiliation{Institute  of  Theoretical  Physics,  TUD Dresden University of Technology, D-01062 Dresden, Germany}
	
	\author{Dario Egloff}
	\email{d.egloff@uniandes.edu.co}
	\affiliation{Institute  of  Theoretical  Physics,  TUD Dresden University of Technology, D-01062 Dresden, Germany}
	\affiliation{Max Planck Institute for the Physics of Complex Systems, N\"othnitzer Strasse 38, 01187 Dresden, Germany}
	\affiliation{Departamento de Física, Universidad de los Andes, Cra 1 No 18A-12, Bogotá, Colombia }
	
	\keywords{Invasive Measurements, Quantumness, Contextuality, Coherence, Non-Classicality, Kolmogorov Consistency Conditions}
	
	\begin{abstract}
In classical stochastic theory, the joint probability distributions of a stochastic process obey by definition the Kolmogorov consistency conditions. Interpreting such a process as a sequence of physical measurements with probabilistic outcomes, these conditions reflect that the measurements do not alter the state of the underlying physical system. Prominently, this assumption has to be abandoned in the context of quantum mechanics, yet there are also classical processes in which measurements influence the measured system. Here, we derive conditions that characterize uniquely classical processes that are probed by a reasonable class of invasive measurements. We then analyse under what circumstances such classical processes can simulate the statistics arising from quantum processes associated with informationally-complete measurements. We expect that our investigation will help build a bridge between two fundamental traits of non-classicality, namely, coherence and contextuality.
	\end{abstract}

\maketitle

\section{Introduction}
Since the inception of quantum physics a very fundamental question driving both its theoretical development and some of its most impressive applications is the difference between this theory and the classical description of the physical world. In recent years, there has been a great advancement in the understanding of two topics at the heart of this question, coherence theory and contextuality (see~\cite{streltsov2017colloquium} and~\cite{Budroni2022} for reviews). 

Coherence theory formalises the intuition that superposition in the number states is a signature of non-classicality~\cite{streltsov2017colloquium,baumgratz2014quantifying}. What started as a parallel development to entanglement theory~\cite{baumgratz2014quantifying,Killoran2016} has since proven useful to develop deep quantitative connections between coherence and a wide range of topics, such as fringe visibility~\cite{Biswas2017,Masini2021}, state and sub-channel discrimination tasks~\cite{napoli2016robustness,Takagi2019,Takagi2019a,Skrzypczyk2019,Skrzypczyk2019a}, power of quantum computation~\cite{hillery2016coherence,matera2016coherent,Ahnefeld2022}, state conversion in the resource theory of thermodynamics~\cite{Lostaglio2015,Lostaglio2019}, quantum discord and entanglement~\cite{Streltsov2015,yadin2016quantum,Ma2016,Egloff2018,hu2018quantum}, quantum steering~\cite{Hu2016}, and, crucially for the present work, non-classical correlations in time~\cite{Smirne2018,strasberg_classical_2019,Milz2020a} 
such as those at the basis of the Leggett-Garg inequalities~\cite{leggett1985quantum}.

Contextuality sprang to live with the no-go theorem of Kochen and Specker, proving that one cannot build a hidden variable theory that assigns truth values to proper finite collections of projective measurements of a quantum system of dimension greater than two~\cite{Kochen1968}. The topic has seen a great development in recent years, for instance, showing how contextuality is a strictly stronger quantum feature than Bell-non-locality~\cite{Fine1982,Liu2016,Cabello2021}, is important for magic state quantum computation~\cite{Raussendorf2013,Howard2014} as well as for quantum channel capacity and quantum state discrimination~\cite{Cubitt2010,Schmid2018}, and is related to non-classical correlations in time~\cite{Szangolies2013,Vieira2022}.

Among the different definitions of (non)contextuality, here we rely on the identification of noncontextual statistical models as those for which there exists a joint probability distribution for all the measurements involved in the statistics \cite{Klyachko2008,Abramsky2011,Chaves2012}, which takes root in the Kolmogorov consistency conditions of the classical statistical theory~\cite{kolmogorov_foundations_1956}. Explicitly, the Kolmogorov consistency conditions state that probabilities are positive, sum to one, and that the joint probabilities satisfy a constraint on the marginalization that reads, taking for simplicity the joint probability associated with two values $x_1$ and $x_2$,
$
    \sum_{x_1} P(x_2, x_1) = P(x_2),
$
where $P(x_2)$ is the probability that the stochastic process assigns to the value $x_2$ only. These conditions are fundamental in physics, because, by virtue of the Kolmogorov extension theorem, they guarantee the existence of an overall classical description of the statistics satisfying them.
In particular, we investigate the multi-time statistics associated with sequential measurements at different times, for which a clear-cut connection has been established between quantum coherence and discord on one side and the breaking of the Kolmogorov consistency conditions in the quantum setting on the other \cite{Smirne2018,strasberg_classical_2019,Milz2020a}; see also the recent review~\cite{Strasberg2023}.

The validity of the Kolmogorov consistency conditions in classical models refers to the possibility, at least in principle, to perform noninvasive measurements that access the actual value possessed by physical quantities without disturbing the subsequent statistics. While this is indeed generally not possible in the quantum realm, as measurements modify the state of the system, it is also true that even classically one can think of measurements modifying the system state. As a specific example, one can think of the measurement of the position of a particle undergoing Brownian motion due to the interaction with, possibly very small, surrounding particles. Such a measurement might in fact modify the positions of all the particles involved and then the following statistics of the particle's position would be different depending on whether the measurement has been performed or not. A visualisation is given in \cref{fig:brownian-motion_inv-meas.}.
\begin{figure}[htb]
    \includegraphics[width=0.4\textwidth]{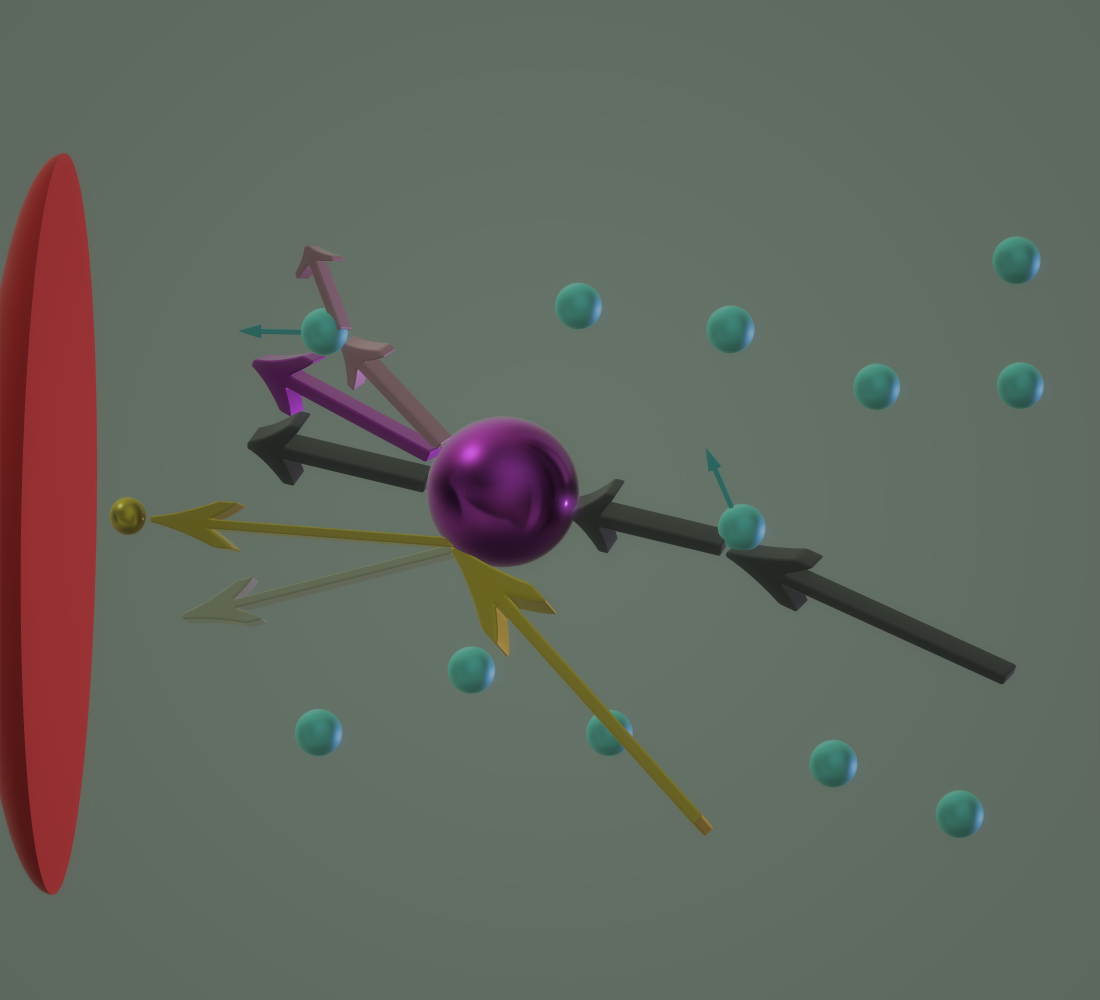}
    \caption{Sketch of a possible setup for an invasive stochastic measurement: A system particle (large violet sphere) moves randomly according to Brownian motion due to collisions with environmental particles (small green spheres). In order to measure its position at a certain time one shoots a smaller probe particle (small gold sphere) at it, which is detected later on a screen. By using the position and angle from which it is shot and the position and angle at which it is detected on the screen one can compute the position of the Brownian particle.
    \newline
    However, due to the impact of the probe particle the system particles momentum is changed. Thus, the probabilities for subsequent positions and paths (violet arrows pointing away from the particle) are altered from what they would be without having performed the position-measurement via the probe particle (black arrows).
    }
    \label{fig:brownian-motion_inv-meas.}
\end{figure}

On the one hand, considering classical invasive measurements opens the door to possible loopholes when trying to certify experimentally the nonclassicality of a given statistics, such as the so-called clumsiness loophole in the context of Leggett-Garg inequalities \cite{Wilde2012}. On the other hand, it allows an extended notion of classicality, where the Kolmogorov consistency conditions no longer hold \cite{Milz2020} and contextuality is accounted for by classical models of invasiveness \cite{Barnett2015,Budroni2022,Vitagliano2023}.

In this paper, we introduce a class of invasive statistical models, starting from a canonical classical process that satisfies the Kolmogorov consistency conditions and including invasiveness via an operational characterisation of the disturbance on the statistics induced by the measurements, along with a restriction on the accessible multi-time probabilities. We derive concrete and experimentally verifiable conditions that uniquely characterise such invasive classical processes, for the case of up to two (informationally complete) measurements and preparations for arbitrary finite-dimensional systems. Furthermore, we provide a general microscopic description of the statistics that satisfy such conditions in terms of a system plus environment model. Lastly, we determine when a quantum statistics can be simulated via the introduced classical invasive model, focusing on informationally-complete POVMs and identifying the key property of the dynamics that is linked to this extended notion of classicality.

\section{Classical invasive model}\label{sec:cim}
In this section, we are going to define what type of invasive models we consider to be classical. This is analogous to the classicality conditions of Bell~\cite{Bell1964}, but here we do not consider local observables. Instead, similarly as in~\cite{leggett_macroscopic_1980}, we consider a system probed multiple times by a measurement. However, differently from \cite{leggett_macroscopic_1980}, we do not consider that our measurement is in some way deterministic, but only ask that it is invasive in a specific way.

To deal with multi-time statistics in the presence of invasive interventions -- such as invasive measurements -- we use the notion of contextual probabilities~\cite{Khrennikov2009}. Intuitively, contextuality means that the observations of an experiment can depend on the setting of the experiment that one does not consider to be part of the actual experiment. In quantum mechanics, for instance, whether one performs a measurement at a given time can affect the outcomes at a later time. In this way, each invasive intervention defines a different context and, within a given context, the Kolmogorov consistency conditions (KCCs) hold. We can restate this using the Kolmogorov extension theorem~\cite{kolmogorov_foundations_1956}: there is a classical stochastic process that gives rise to the observable probabilities if one does not change the context. However, the collective probability distribution does not respect the KCC; taking the marginal over the outcomes of one invasive measurement does not generally tell us what would happen if that specific measurement was not performed. Again restating this last sentence using the Kolmogorov extension theorem, we find that, in general, there is no classical stochastic process that gives rise to the observable probabilities if one changes the context.

Note that contextual probabilities provide a general formalism to describe contextual theories and, if defined broadly enough, they encompass all the predictions of quantum mechanics~\cite{Khrennikov2009}. 
In the following, however, we are going to use them to define a specific class of contextual theories, 
which we still understand as classical in view of the kind of invasiveness allowed.

Besides measurements, we describe explicitly the possibility to re-prepare the system after any measurement and before the following evolution (and further measurement); indeed, also such an intervention can, in general, be invasive. Note that the formalism can be used even if no re-preparations are done. Hence, in the framework we consider here, contexts are defined by the sequence of invasive measurements chosen at different times, together with the possible re-preparation of the system after the measurements. In each context there are non-invasive measurements (at different times) that correspond to the canonical classical model, where probability distributions referring to the same sequence of times are related by KCCs. The basic idea is that such ideal measurements cannot be performed, so that only probabilities involving invasive measurements and re-preparations are actually observable. On the other hand, invasive measurements can be characterised with respect to the non-invasive ones, which defines the classical invasive model at hand. In this way, we separate explicitly the part of the model that is affected by the invasiveness of the measurement and preparation, from the part that is thought to be only due to the dynamical evolution of the system.

\subsection{Instantaneously-invasive measurements} 

Our first aim is to derive consistency conditions referring to different contexts, characterized by which invasive measurements and state re-preparations are performed. Our model includes the possibility of re-preparing the system after each invasive measurement, such that a non-invasive measurement would give a definite outcome with certainty. This means that we can have optimal control of the states even after the invasive measurements. As we assume that the re-preparation is deterministic, we condition on the specific choice, in order not to carry  in the statistics information that depends solely on a fully controllable choice in the experiment. All in all, we will consider probabilities of the form $P^{R_{n-1}, A_{n-1}; \ldots; R_1, A_1}(a_n,\ell_n; \ldots; a_1, \ell_1|r_{n-1};\ldots ;r_1)$, that is, the hypothetical probability of getting outcomes $\left\{\ell_1, a_1\right\}$ for a non-invasive measurement $L_1$ (that cannot be performed in a real experiment) followed by the invasive one $A_1$ at time $t_1$ and so on until $\left\{\ell_{n}, a_{n}\right\}$ at time $t_{n}$, with $t_n\geq\ldots\geq t_1$, conditioned on the re-preparations $R_1, \ldots R_{n-1}$ of the system in, respectively, $r_1, \ldots, r_{n-1}$ instantly after $A_1, \ldots, A_{n-1}$. As we are interested in the effect of the measurements, we only consider the possibilities of either a specific measurement $A_i$ (or re-preparation $R_i$) being performed at time $t_i$ or not. In the second case, the letter $A_i$ is omitted in the superscript of $P$. \emph{Whenever $A_1,R_1, \ldots, A_{n-1}, R_{n-1}$ are fixed}, the standard KCCs apply, so that we have, for example,
$
\sum_{\ell_1}P^{R_1,A_1}(a_2,\ell_2;  a_1, \ell_1|r_1)
=P^{R_1, A_1}(a_2,\ell_2;a_1|r_1),
$
while, in general,
$
\sum_{a_1} P^{A_1}(a_2; a_1) \neq P(a_2),
$
since the probabilities at the left and right hand sides of the previous expression refer to two different contexts, one where the invasive measurement $A_1$ at time $t_1$ is performed and the other where it is not. Also note that we write $P^{R_{n-1}, A_{n-1}; \ldots; R_1, A_1}(a_n,\ell_n; \ldots; a_1, \ell_1|r_{n-1};\ldots ;r_1)$ rather than $P^{A_n;R_{n-1}, A_{n-1}; \ldots; R_1, A_1}(a_n,\ell_n; \ldots; a_1, \ell_1|r_{n-1};\ldots ;r_1)$, that is, we do not define a context for the last measurement $A_n$. This simplification can be done assuming causality (see~ \cref{ass4}), as there is no subsequent probability that could depend on whether or not the last measurement was performed. Consequently, there is no need to define a context for the last measurement.

We now enumerate the conditions that characterise the class of invasive theories we take into account.  Firstly, we specify the operational definition of the invasive measurements, in terms of the hypothetical non-invasive measurements.\\
\begin{ass}\label{ass1}
    \emph{Whatever the previous (or subsequent) sequence of measurements,} if a non-invasive measurement $L$ would give the outcome $\ell$ with certainty, then  the probability to get an outcome $a$ when performing instead an invasive measurement $A$ is
    \begin{equation}
        \textnormal{Prob}(\textnormal{inv } A\, \mapsto a | \textnormal{non-inv } L\, \mapsto \ell) = M_{a; \ell}.
        \label{eq:pma}
    \end{equation}
\end{ass}
The so defined matrix $M$ will be called \textit{Invasive Measurement Matrix} or just \textit{IMM} and it is indeed a stochastic matrix.

A practical interpretation of this condition is the following. Assume for a moment that in principle we were able to perform a hypothetical non-invasive measurements on the system. Further assume that if we prepared the system in a fixed way, when measuring we always get the same outcome $\ell$ determined by the system state. In such a situation, the condition states that the probability of the invasive measurement to result in outcome $a$ will be $P(a) = M_{a; \ell}$. This means that while the invasive measurement measures the same physical quantity as a hypothetical non-invasive one, it might disturb the measured state. Therefore, even a well-defined state will not necessarily give always the same outcome, when measured by such an invasive measurement.

More specifically, \cref{ass1} implies that we consider invasive theories where the influence of the measurement on the subsequent statistics is ``instantaneous", i.e., it does not depend on the previous (neither on the following) sequence of measurement outcomes. This is indeed a fully motivated restriction from a physical point of view and can be seen as the counterpart of the use of a sequence of quantum instruments to describe subsequent measurements on a quantum system (see also the next section). 

A simple example where this condition is not satisfied, is the following. Suppose that the same measurement device is used in two subsequent measurements and in the second one the device does not measure the system at all, but simply shows the outcome of the first measurement. In this case, the probability distribution of the second measurement will depend on the first one and in general cannot be described solely by the state of the system just before the second measurement. This condition thus formalises the confidence of a careful experimentalist that such unwanted dependencies between measurements do not happen in the experiment. 

Explicitly, \cref{ass1} means that $\forall k=2,\ldots n, $
\begin{align}
    &P(a_k, \ell_k) = M_{a_k; \ell_k} P(\ell_k), \label{eq:p1}\\
    &P^{R_{k-1},A_{k-1}; \ldots; R_1,A_1}(a_k, \ell_k; \ldots ;a_1, \ell_1|r_{k-1};\ldots;r_1) \label{eq:p2}
    \\
    &=M_{a_k; \ell_k} P^{R_{k-1}, A_{k-1}; \ldots; R_1, A_1}(\ell_{k}; \ldots; a_1, \ell_1|r_{k-1};\ldots;r_1)
    \nonumber
\end{align}
which clarifies the role of $M_{a_k, \ell_k}$ as the conditional probability relating a sequence of measurements ending with $\ell_k$ with the one obtained by adding $a_k$. The IMM $M$ can be fully reconstructed from the probabilities associated with the invasive measurement. In fact, if we can prepare the system in a way such that a subsequent non-invasive measurement (for example, at time $t_1$) would result in $P(\ell_1) = \delta_{\ell_1, \overline{\ell}}$ for any of the possible outcomes $\overline{\ell}$, \cref{eq:p1} then gives us 
\begin{equation}\label{eq:ma1}
    M_{a_1; \overline{\ell}} = P(a_1| \ell_1=\overline{\ell}),
\end{equation}
i.e., $M_{a_1; \ell_1}$ can be reconstructed by preparing the state $\ell_1$ and registering the probability associated with the subsequent invasive measurement with outcome $a_1$.

The second condition about the invasive measurement concerns its \emph{completeness}.\\
\begin{ass}\label{ass2}
	The invasive measurement is informationally complete (IC) (but not over-complete), i.e.,
	$\left\{P(a)\right\}_{a}$ allows us to infer the one-step statistics of \emph{any}
	other measurement (invasive or not) performed at the same time.
\end{ass}

The informational completeness of the measurements, both invasive ones as well as the (actually not performable) non-invasive ones, implies that the probability distributions $P^{A_k}(a_k)$ and $P(\ell_k)$ both represent the same abstract underlying state - in the first case measured invasively and in the second case hypothetically non-invasively. This implies a one-to-one correspondence, i.e. a bijection, between both representations given by the IMM $M$ and thus $M$ must be invertible. This condition basically says that while our measurements are not ideal, in the sense that they do alter the measured system, they at least give us full information. Informational completeness corresponds to the well-known situation in quantum mechanics where one performs a (minimal) tomography: any prediction of the statistics of any possible measurement can be inferred from that information. In future work, we plan to investigate what happens when such an assumption is weakened to include, on the one hand, also the case of orthogonal projective measurements in quantum mechanics (which are not complete), or, on the other hand, situations where the measurements are over-complete. Over-complete just means that the measurement gives at least the necessary information to recover the state.

The next two conditions characterize the influence that re-preparing the system may have on the statistics. The first one fixes the interplay between the different interventions (invasive and non-invasive ones), in this way connecting probabilities referring to different contexts.\\
\begin{ass}\label{ass3}
    Given a sequence of a non-invasive measurement, an invasive one and a re-preparation procedure, all at the same time, the invasive measurement does not affect the subsequent statistics; explicitly, if said sequence occurs at all times $t_1, \ldots t_{n-1}$, one has
    \begin{align}
        &P^{R_{n-1}, A_{n-1}; \ldots; R_1, A_1}(\ell_n ;
        \ldots;a_1, \ell_1| r_{n-1};\ldots;r_1)  \label{eq:ecc} \\
        &=M_{a_{n-1};\ell_{n-1}}\ldots M_{a_{1};\ell_{1}} P^{R_{n-1}; \ldots; R_1}(\ell_n ; 
        \ldots; \ell_1|r_{n-1};\ldots;r_1).\nonumber
    \end{align}
\end{ass}
Intuitively, after a system is re-prepared to a given state, the evolution should not depend on the outcome of the measurement before the re-preparation. That is because the outcome has been discarded in the re-preparation. \cref{ass3} reflects this thought. However, the condition is not trivial. Essentially, it means that the invasive measurement only affects the degrees of freedom of the measured system. We will come back to this after introducing a picture of the statistics based on the interaction of the measured system with an environment.

The next condition concerns our ability to prepare the system in a way such that if we do not alter the system's state (that has a meaningful definition due to~\Cref{ass2}), we do not alter the subsequent evolution.
\begin{ass}\label{ass0}
	The statistics stemming from re-preparing the system in a state labeled by $a_1$ after getting the measurement outcome $a_1$ cannot be distinguished from only measuring $a_1$; for instance:
	\begin{align}\label{cond:reprep}
	P^{A_1}(a_2;a_1)
	=P^{R_1,A_1}(a_2;a_1|r_1=a_1).
	\end{align}
\end{ass}
To be explicit, this condition is in general not valid, if the re-preparation affects the setting of the experiment. In other words, it formalises the experimentalist capability to affect only the intended degrees of freedom in the re-preparation.

Finally we also assume the following.
\begin{ass}\label{ass4}
	Actions later in time do not affect earlier actions, meaning that one can always take the marginal over later actions to get the former ones; for instance:
	\begin{align}\label{cond:causal}
		\sum_{\ell_2} P^{R_1}(\ell_2;\ell_1|r_1)=P(\ell_1) \; \forall r_1.
	\end{align}
\end{ass}
This condition is nothing else than the causality condition, which is usually required in general probability theories, and has been named arrow-of-time condition in the framework of sequential measurements at different times \cite{Clemente2016}.

Although some of the conditions presented here might seem trivial, we need to assume them explicitly. This is because the object we want to analyse is the statistics, that per se does not need to satisfy any of the conditions: in fact, we do not assume any background theory that may have some of these conditions already incorporated. Instead, given a statistics, we ask whether it is possible that it stems from a classical invasive process with instantaneous invasive and informationally-complete measurements. Hence, we ask whether the given statistics fulfills all the conditions presented above and is consistent with a classical stochastic process. Moreover, we stress that the class of invasive theories defined by the conditions above does not cover any possible description that can be considered as a classical simulation of temporal correlations appearing in quantum mechanics. In particular and quite significantly, compared for example to the approach put forward in \cite{Fritz2010,Hoffmann2018}, in the classical invasive models we consider invasiveness is fully encoded into the IMM, whose dimensionality is limited by the number of the possible measurement outcomes.

\subsection{Necessary and sufficient conditions for the existence of the invasive-measurement description}
Before proceeding, we clarify which probabilities are directly accessible in the invasive theories we describe. From here on, we restrict ourselves to the case where one can perform (or not) invasive measurements $A_1$ and $A_2$ only at two different fixed instants of time $t_1$ and $t_2\geq t_1$, i.e., $n=2$, leaving for future investigation the extension to a generic number $n$ of invasive measurements at $n$ subsequent times.

The very notion of invasive theory we are using means that statistics referring to non-invasive measurements cannot be accessed directly, so that, for example, one cannot obtain the probability $P^{R_1,A_1}(a_2,\ell_2; a_1, \ell_1|r_1)$ from empirical data. On the contrary, one can indeed access the probabilities where only invasive measurements are involved, such as $P(a_i)$ and $P^{A_1}(a_2 ; a_1)$. In addition, also probabilities involving only invasive measurements and re-preparations can be accessed, as in $P^{R_1,A_1}(a_2; a_1|r_1)$, that is the probability of getting the outcome $a_2$ for an invasive measurement at time $t_2$ and $a_1$ for an invasive measurement at time $t_1$, conditioned on having re-prepared the system in the state $r_1$ after the first invasive measurement. Analogously, we can also simply perform a state preparation at time $t_1$, but without any invasive measurement at that time, in this way accessing $P^{R_1}(a_2| r_1)$.

Finally, we stress that there are probabilities that cannot be accessed directly, but that can still be \emph{reconstructed} from observable probabilities. Indeed, $P(\ell_1)$ is an example of this due to the~\cref{ass2}; denoting as $(M^{-1})_{\ell; a}$ the matrix elements of the inverse of $M$, we have in fact (from the sum over $\ell_1$ of \cref{eq:p1}) 
\begin{equation}\label{eq:l1}
    P(\ell_1) = \sum_{a_1}(M^{-1})_{\ell_1; a_1}P(a_1),
\end{equation} 
i.e., the statistics associated with the non-invasive measurement at time $t_1$ can be inferred from the statistics associated with the invasive measurement at the same time. Note that the possibility to do so depends on the invertibility of the IMM $M$ and is thus a consequence of the informationally completeness of the measurement. Similarly, we have $$P^{R_1,A_1}(\ell_2; a_1|r_1) = \sum_{a_2}(M^{-1})_{\ell_2; a_2}P^{R_1,A_1}(a_2;a_1|r_1).$$

\cref{ass3} is the key element that allows us to connect probabilities referring to different contexts, i.e., to situations where there is or there is not the intermediate invasive measurement at time $t_1$. In particular, \cref{eq:ecc} for $n=2$, along with $\sum_{a_1}M_{a_1, \ell_1}=1$ ($M$ is a stochastic matrix), \cref{eq:p2} and the KCCs with respect to $\ell_1$ and $\ell_2$ imply
\begin{align}
    P^{R_1}(a_2| r_1) = \sum_{a_1}P^{R_1,A_1}(a_2;a_1|r_1), \label{eq:c1}
\end{align}
that is, one can apply the standard KCC with respect to $A_1$ when both the invasive measurement and the re-preparation are involved at time $t_1$. Moreover, as shown in \cref{app:the}, \cref{eq:ecc} also implies
\begin{align}
    P(a_2) = \sum_{a_1,r_1}(M^{-1})_{r_1;a_1}P^{R_1,A_1}(a_2; a_1|r_1). \label{eq:c2}
\end{align}
Crucially, these relations involve only probabilities that are referring to invasive measurements and a state re-preparation and are thus accessible (see the remark at the beginning of this paragraph and \cref{eq:ma1} for the assessment of $M$).

We have now all the ingredients we need to formulate the first main result of the paper.
\begin{thm}\label{pr:pr}
    Let the probabilities $P(a_1), P(a_2),$ $P^{R_1}(a_2| r_1)$ and $P^{R_1,A_1}(a_2; a_1|r_1)$, as well as an invertible matrix of transition probabilities $M$, be given. Then, i) $P(\ell_1):=\sum_{a_i}(M^{-1})_{\ell_i; a_i}P(a_i)$ and  $P^{R_1}(\ell_2;\ell_1|r_1):=\sum_{a_1,a_2} (M^{-1})_{\ell_2;a_2}  (M^{-1})_{\ell_1;a_1} P^{R_1,A_1}(a_2;a_1|r_1)$ are probability distributions, and ii) \cref{cond:reprep,cond:causal,eq:c1,eq:c2} hold, if and only if there exists a probability distribution $P^{R_1,A_1}(a_2, \ell_2; a_1, \ell_1|r_1)$, from which the probabilities above can be obtained by \cref{ass0,ass1,ass2,ass3,ass4} together with the KCCs over the corresponding $\ell_i$s. 
\end{thm}
In other terms, we have some definite conditions on experimentally accessible probabilities that, if satisfied, guarantee the existence of an underlying contextual model that accounts for the given statistics and, by satisfying \cref{ass0,ass1,ass2,ass3,ass4}, describes instantaneously-invasive informationally-complete mea\-sure\-ments.

Quite interestingly, the proof of the statement, see \cref{app:the},  is constructive and consists in the introduction of further degrees of freedom (an environment) interacting with the system the statistics is referring to. As depicted in \cref{fig:path}, the global evolution on the system together with the environment can be modelled as a stochastic evolution, where one needs to account for the stochastic intervention of the measurements on the system whenever one performs them. As shown in \cref{app:the}, such a model exists, whenever the conditions i) and ii) stated in the theorem are satisfied. It is then easy to verify that the reduced dynamics on the system gives a contextual model that describes instantaneously-invasive informationally-complete mea\-sure\-ments, reproduces the statistics and satisfies the conditions stated in the theorem. Thus, the equivalence of the two models with conditions i) and ii) of the theorem is shown. 

As a further remark, we comment on the question, ``What if the statistics stems from an experiment where no re-preparations have been done?" In this case, one can still apply the theorem to tell whether the statistics could be reproduced by a classical stochastic process probed by instantaneously-invasive informationally-complete measurements. One still just needs to check \cref{cond:reprep,cond:causal,eq:c1,eq:c2} together with the KCCs over the corresponding $\ell_i$s for the statistics one has. However, \cref{cond:reprep} is trivial in that case and \cref{cond:causal,eq:c1} get substantially weakened (as one cannot check whether they are fulfilled for any re-preparation). However one could still check whether the probabilities without re-preparation can be embedded in a statistics including re-preparation and such that the assumptions of the Theorem hold.

Finally, note that the proof also shows that if $P(\ell_1)$ and $P^{R_1}(\ell_2;\ell_1|r_1)$ are quasi probability distributions (i.e., they can have negative entries), the theorem holds up to $P^{R_1,A_1}(a_2, \ell_2; a_1, \ell_1|r_1)$ having negative entries.

\begin{figure}[htb]
	\def\svgwidth{\linewidth}
\begingroup%
  \makeatletter%
  \providecommand\color[2][]{%
    \errmessage{(Inkscape) Color is used for the text in Inkscape, but the package 'color.sty' is not loaded}%
    \renewcommand\color[2][]{}%
  }%
  \providecommand\transparent[1]{%
    \errmessage{(Inkscape) Transparency is used (non-zero) for the text in Inkscape, but the package 'transparent.sty' is not loaded}%
    \renewcommand\transparent[1]{}%
  }%
  \providecommand\rotatebox[2]{#2}%
  \newcommand*\fsize{\dimexpr\f@size pt\relax}%
  \newcommand*\lineheight[1]{\fontsize{\fsize}{#1\fsize}\selectfont}%
  \ifx\svgwidth\undefined%
    \setlength{\unitlength}{283.46456693bp}%
    \ifx\svgscale\undefined%
      \relax%
    \else%
      \setlength{\unitlength}{\unitlength * \real{\svgscale}}%
    \fi%
  \else%
    \setlength{\unitlength}{\svgwidth}%
  \fi%
  \global\let\svgwidth\undefined%
  \global\let\svgscale\undefined%
  \makeatother%
  \begin{picture}(1,0.54359327)%
    \lineheight{1}%
    \setlength\tabcolsep{0pt}%
    \put(0,0){\includegraphics[width=\unitlength,page=1]{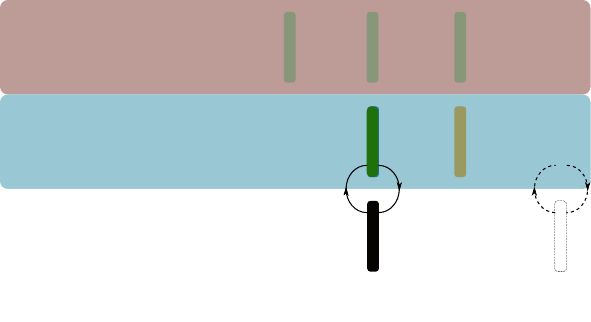}}%
    \put(-0.00233334,0.42462134){\color[rgb]{0,0,0}\makebox(0,0)[lt]{\lineheight{1.25}\smash{\begin{tabular}[t]{l}$E$\end{tabular}}}}%
    \put(-0.00025322,0.26667776){\color[rgb]{0,0,0}\makebox(0,0)[lt]{\lineheight{1.25}\smash{\begin{tabular}[t]{l}$S$\end{tabular}}}}%
    \put(-0.00217029,0.11402578){\color[rgb]{0,0,0}\makebox(0,0)[lt]{\lineheight{1.25}\smash{\begin{tabular}[t]{l}$R$ \end{tabular}}}}%
    \put(0,0){\includegraphics[width=\unitlength,page=2]{path_total_new.pdf}}%
  \end{picture}%
\endgroup%

	\caption{This figure shows the stochastic model of system $S$ plus environment $E$ that can explain the statistics gathered in the register $R$, exactly if there is a corresponding contextual model that describes instantaneously-invasive informationally-complete mea\-sure\-ments. The model starts with a hypothetical system state $s_0$ at time $t_0$ that evolves to the (not measurable) state $\ell_1$ on the system and $e_1$ on the environment under the action of the stochastic map $V_{\ell_1,e_1;s_0}$ at time $t_1$. Then there is a measurement on the system side, which changes the system's state under the action of the stochastic map $M_{a_1;\ell_1}$ to $a_1$. If the system is re-prepared in the state $r_1$, the state $a_1$ is lost. The evolution then continues analogously during time $t_2$. The model is able to reproduce the statistics if the stochastic matrices of the evolutions $V$ are independent of whether one measures or re-prepares the system at any time.}
    \label{fig:path}
\end{figure}

\section{Quantum processes with informationally-complete quantum measurements}
In the following we discuss the application of the concepts developed above to the statistics of quantum sequential measurements at different times, i.e., we investigate to what extent the predictions of quantum mechanics can be reproduced via classically invasive models as those defined in the previous section. 

In contrast to the classical  case, mutually exclusive outcomes in quantum mechanics -- i.e. an orthogonal measurement set-up -- cannot reveal the full information about the state of a quantum system, and informationally complete quantum measurements have overlapping outcomes. The idea is then to interpret this overlap as stochastic invasiveness of the kind introduced above. Consequently, we will define conditions for a quantum stochastic process associated with informationally complete measurements to provide statistics that obey the properties and conditions of consistency as discussed above. Furthermore, we will connect the fulfillment of such conditions to a definite property of the evolution of a quantum system interacting with an environment and realizing the process at hand.

\subsection{IC-POVM}
We start by recalling the definition of informationally complete quantum measurements~\citep{Renes2004,Kovacevic2008}. A rank-one \textit{Informationally Complete Positive Operator Valued Measure} (IC-POVM) is a set of positive operators $\{\E_\psi \ = \ K_\psi^\dagger K_\psi \ = \ \frac{1}{c_\psi} \ketbra{\psi}{\psi} \}$ where $\{ \ketbra{\psi}{\psi} \} =: \F$ is a frame on the space of bounded operators $\B(\Hil)$  (called a \textit{quantum frame} \citep{Renes2004,Richter2022}), i.e. any operator, like density operators, have a decomposition 
$$
\Rho \ = \ \sum_\psi f_\psi \ketbra{\psi}{\psi},
$$ 
and $c_\psi$ is chosen such that $\sum_\psi \E_\psi = \id$. For simplicity we assume the quantum frame (and IC-POVM) to be minimal, that is, a (non-orthogonal) basis. 

Under this assumption, the \textit{frame decomposition coefficients} (FDCs) $f_\psi$ of density operators representing mixed quantum states are $f_\psi \in \mathbb{R}$ due to hermiticity of $\Rho$ and $\sum_\psi f_\psi = 1$ since $\Rho$ is trace-one. Thus, using a quantum frame one can express any quantum state $\Rho$ as an at least quasi-stochastic mixture of a \emph{fixed set of pure quantum states} $\{\ketbra{\psi}{\psi}\}$. In the case of an open quantum system $\Hils$ coupled to an environment $\Hile$, i.e. giving a global space $\Hil = \Hils \otimes \Hile$, it is possible to combine a system quantum frame $\F_\sys := \{\ketbra{\psi}{\psi}_\sys \}$ and an environmental frame $\F_\env := \{\ketbra{\epsilon}{\epsilon}_\env \}$ to an overall multi-partite quantum frame $\F := \F_\sys \otimes \F_\env = \{ \ketbra{\psi}{\psi}_\sys \otimes \ketbra{\epsilon}{\epsilon}_\env \}$ s.t.
\begin{align}\label{eq:FDC}
    \Rho &= \sum_{\psi, \epsilon} f_{(\psi, \epsilon)} \ketbra{\psi}{\psi}_\sys \otimes \ketbra{\epsilon}{\epsilon}_\env = \sum_\psi f_\psi^\sys \ketbra{\psi}{\psi} \otimes \epsilonbm_\psi
\end{align}
where $f_\psi^\sys$ are the FDCs of $\Rho_S:=\partTr{\env}{\Rho}$ and the operators $\epsilonbm_\psi \in \B(\Hile)$ are hermitian, trace-one but not necessarily positive semi-definite. Consider performing a measurement, given by an IC-POVM on $\Hils$. If the system $\Hils$ before the measurement was in the state $\Rho$ and the outcome  is $\psi$, the state after the measurement is given by 
\begin{equation}\label{eq:kappa}
    \K_\psi(\Rho) = (K_\psi\otimes\id) \Rho (K_\psi\otimes\id)^\dagger.
\end{equation} 
As in this example we are considering the Kraus operators to be rank one, the state after the measurement can also be expressed as
\begin{equation}\label{eq:pipsi}
    \Pi_\psi (\Rho) = \ketbra{\psi}{\psi}_\sys \otimes \Rho_\env.
\end{equation}
Moreover, we consider an intermediate evolution between subsequent measurements, as given by a unitary map $\Evl$ acting on system and environment ($\Evl(\Rho) = V \Rho V^\dagger$, with $V \in \mathcal{U}(\Hil)$ a unitary operator). Let us assume, without loss of generality, that the input state at time $t_1$ is generated by a unitary $\Evl_0$ out of an initial state $\Rho_0$ at a time $t_0$, while the unitary between time $t_1$ and $t_2$ is denoted by $\Evl_1$. That we assume an already evolved state $\Evl_0(\Rho_0)$ entering the very first measurement is motivated by typical assumption from open quantum theory. Using this construction we can initially assume e.g. product states $\Rho_0 = \Rho_\sys \otimes \tau$ with a certain system state $\Rho_\sys$ coupled to some thermal state $\tau$ of some bath as environment and nevertheless allow for entering arbitrarily correlated states at the moment of the first measurement.

\subsection{Correspondence between classical invasive and quantum models}
All probabilities that define the observable quantities in the invasive model defined in \cref{sec:cim} can be expressed via the Born's rule applied to the proper sequence of maps -- note that the indices $a_i,  \ell_i, r_i$ now refer to the elements of the quantum frame indexed by $\{\psi\}$, where $\ketbra{\psi}{\psi}$ labels the projectors defining the IC-POVM. For instance,
\begin{align}
    P^{R_1, A_1}(a_2; a_1|r_1)  &= \Tr{\K_{a_2} \Evl_1 \Pi_{r_1} \K_{a_1}\Evl_0(\Rho_0)}\nonumber
    \\
    &= \frac{ \Tr{\K_{a_2} \Evl_1 \Pi_{r_1} \K_{a_1} \Evl_0(\Rho_0)}}{\Tr{\Pi_{r_1} \K_{a_1}\Evl_0(\Rho_0)}} \Tr{\K_{a_1} \Evl_0(\Rho_0)},\label{eq:pr1a1}
\end{align}
where $\K_{a}$ describes the state transformation due to a measurement with outcome $a$ according to \cref{eq:kappa}, while $\Pi_{r}$ is the re-preparation in the state one gets after a measurement with outcome $r$ according to \cref{eq:pipsi}. From $P(a_1)  =  \Tr{\K_{a_1} \Evl_0(\Rho_0)}$ and setting 
\begin{equation}
    M_{a;\ell} := \Tr{\E_a \ketbra{\ell}{\ell}} = \Tr{K_a \ketbra{\ell}{\ell} K_a^\dagger}
\end{equation} 
one can derive (see \cref{app:pol})
\begin{align}
    P(\ell_1) = f^\sys_{\ell_1}.
\end{align}
Thus, the (inaccessible) probability of measuring outcome $\ell_1$ in a (hypothetically) non-invasive measurement $L_1$ (see ~\cref{eq:l1}) in the invasive-stochastic model is given by the FDC $f^\sys_{\ell_1}$ of the reduced system state (see \cref{eq:FDC}). The frame decomposition at time $t_1$ reads
\begin{align}
    \Evl_0 (\Rho_0) \ := \ \sum_{(\psi, \epsilon)} \left(V_0 \vec{f}_0\right)_{(\psi, \epsilon)} \ \ketbra{\psi}{\psi}_\sys \otimes \ketbra{\epsilon}{\epsilon}_\env,
\end{align}
where we understand $\vec{f}_0$ as a vector like representation of $\Rho_0$ based on its FDCs and $V_0$ accordingly as a matrix like representation of $\Evl_0$ as similarly suggested e.g. by \citep{Yashin2020,Kiktenko2020}. Consequently, we will neglect the indices $\sys$ and $\env$ and use the convention that the first factor of a tensor product refers to $\Hils$ and the second one to $\Hile$. In \cref{app:pol}, the following lemma is shown.
\begin{lem}\label{lem:consistency}
    A quantum stochastic process using IC-POVMs with probabilities as defined above fullfils \cref{cond:reprep,cond:causal,eq:c1,eq:c2}. Furthermore, $P^{R_1}(\ell_2;\ell_1|r_1):=\sum_{a_1,a_2} (M^{-1})_{\ell_2;a_2}(M^{-1})_{\ell_1;a_1} P^{R_1,A_1}(a_2;a_1|r_1)$ and $P(\ell_1)$ are quasi probability distribution (they sum to one, but are not necessarily positive). 
\end{lem}

\Cref{lem:consistency}shows that a probability distribution produced by such a quantum process using IC-POVMs is at least quasi-stochastic, i.e. consistency holds and even objects like $P(\ell_1)$ are real and sum up to one but might be negative. To characterise the cases in which all entities really behave like proper -- positive -- probabilities, we introduce the following definitions.

\begin{defin}\label{def:F-separability}
    A quantum state $\Rho$ on $\Hil = \Hils \otimes \Hile$ is called \textbf{$\F_\sys$-separable} if and only if it has a decomposition $\Rho = \sum_\psi f_\psi^\sys \ketbra{\psi}{\psi} \otimes \epsilonbm_\psi$ with $f_\psi^\sys \geq 0$ and $\epsilonbm_\psi$ is a proper environmental quantum state $\forall \ketbra{\psi}{\psi} \in \F_\sys$. A unitary evolution $\Evl \in \mathcal{U}(\Hil)$ is called \textbf{$\F_\sys$-separable} if and only if it maps $\F_\sys$-separable states to such states again.
\end{defin}
The following theorem, which is proved in \cref{app:th3}, characterises $\F_\sys$-separability as the key property that allows us to reproduce the predictions of quantum mechanics via the classical invasive models introduced in the previous section.
\begin{thm}\label{th:th}
    A quantum process using an $\F_\sys$-based IC-POVM on $\Hils$ as measurement,  $\F_\sys$-separable initial state $\Rho_0$ and $\F_\sys$-separable unitaries $\Evl_{0}, \Evl_1 \in \mathcal{U}(\Hil)$ as initial and intermediate evolutions produce a proper stochastic probability distribution for all contexts.
\end{thm}

\subsection{Markovian and non-Markovian processes}
\cref{th:th} states that a process at hand can be simulated via invasive stochastic probabilities whenever the condition of $\F_\sys$-separability is ensured. Analogously to what happens in the case of ideal projective measurements~\cite{Smirne2018,Milz2020}, there is an important class of processes for which $\F_\sys$-separability reduces to a simpler condition, expressed in terms of the dynamical maps acting on the open system only; namely, this is the case for  Markovian processes. 

Here, what we mean with Markovianity is that the whole hierarchy of probability distributions, and hence in particular the probabilities involved in our analysis, is fixed by the completely positive trace preserving (CPTP) dynamical maps between two subsequent measurements $i$ and $i+1$ defined as
\begin{equation}\label{eq:redmap}
    \Lambda_{i}(\Rho) = 
    \partTr{\env}{\Evl_{i} (\Rho \otimes \hat{\tau}_i) \Evl_{i}^\dagger}
\end{equation}
where $\hat{\tau}_i$ is a reference state of the environment (possibly different at different times). Hence, Markovianity is here understood in terms of a property of multi-time probability distributions, analogously to the definition for classical stochastic processes; for a comparison among different notions of quantum Markovianity, we refer the reader to \cite{Li2018}. This setting means that for any measurement time all relevant information for the subsequent statistics is stored in the system state $\Rho_\sys \in \B(\Hils)$ only and can hence be encoded in a simple frame vector $\vec{f}$ for the system frame $\F_\sys$ corresponding to the IC-POVM at hand. In turn, the frame-representation of a CPTP map is simply a matrix $V_\Lambda$ which maps the frame vector of the input state to the frame vector of the output state. As a consequence, $V_\Lambda$ has to be a quasi-stochastic matrix, i.e. all entries are real and each column sums up to one. 

Now, if a quantum state has non-negative FDCs only (i.e. $\vec{f}$ has non-negative entries), we say that it is \textit{$\F_\sys$-positive} and, accordingly, we define $\F_\sys$-positivity of a CPTP map by requiring that it maps $\F_\sys$-positive states to $\F_\sys$-positive states again; indeed, this is equivalent to the requirement that the corresponding $V_\Lambda$ is a proper stochastic matrix (all entries are non-negative). Even more, when the quantum process is Markovian, this is enough to ensure the simulability via invasive processes. It is in fact easy to see that requiring an $\F_\sys$-separable initial quantum state $\Rho_0$ in \cref{th:th} reduces for product states to the necessity of an $\F_\sys$-positive quantum state $\Rho_\sys$ on the system side and that $\F_\sys$-positivity is the Markovian reduction of $\F_\sys$-separability of \cref{def:F-separability}. An equivalent characterisation of $\F_\sys$-positivity is that such a state $\Rho$ is in the convex hull $\spancvx{\F_\sys}$ of the frame and that such a CPTP map sends its convex hull into its convex hull again. For an illustration of a convex hull of a quantum frame for a qubit see \cref{fig:SIC-tetrahedron}. Thus, the probability distribution of any quantum Markovian process using an $\F_\sys$ based IC-POVM, an initial state $\Rho \in \spancvx{\F_\sys} $ and an intermediate CPTP map $\Lambda: \spancvx{\F_\sys} \to \spancvx{\F_\sys}$ fulfils all conditions for a stochastically invasive statistic.
\begin{figure}[htb]
    \includegraphics[width=0.4\textwidth]{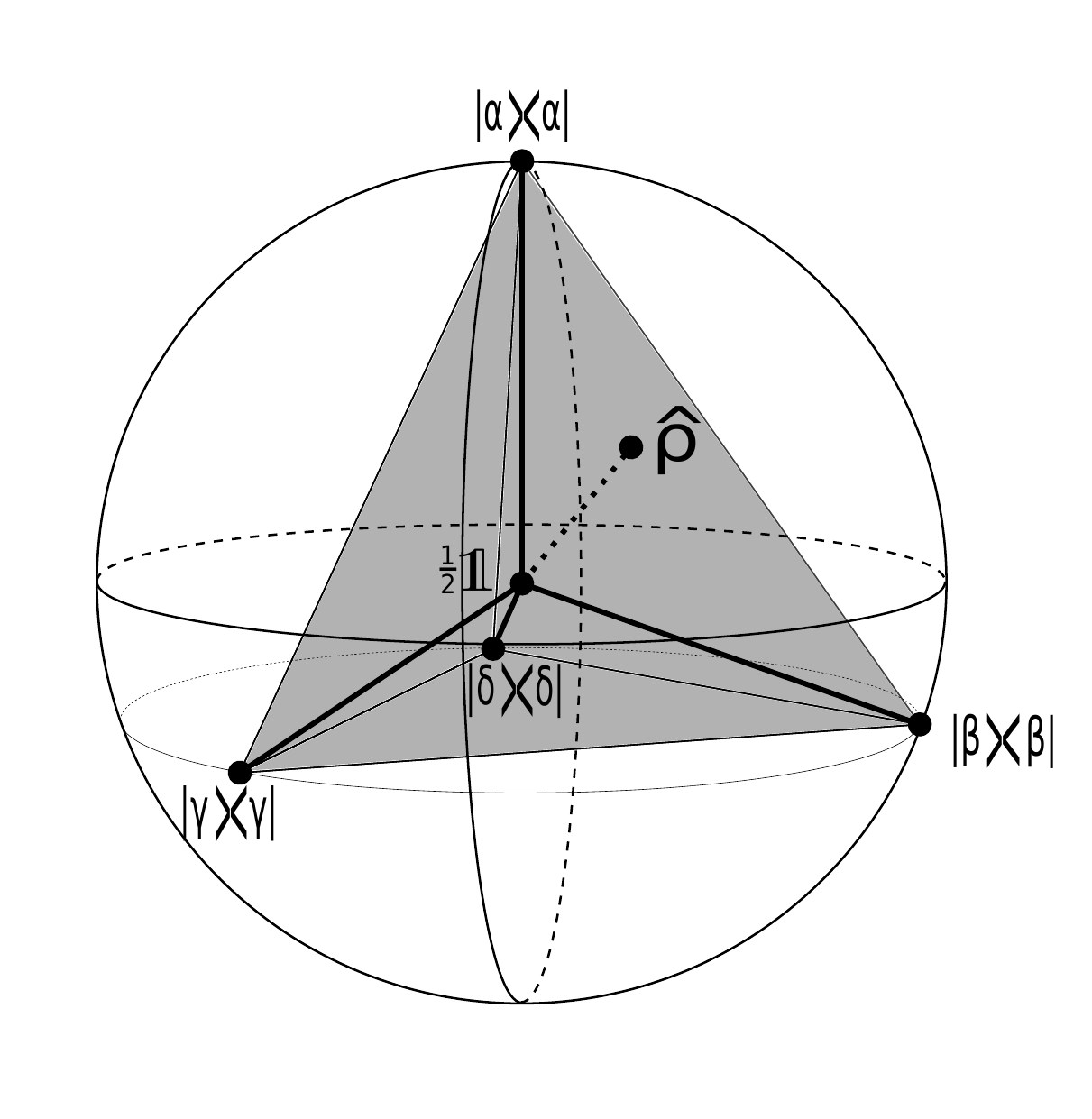}
    \caption{A quantum frame $\F_\mathrm{SIC} = \{\ketbra{\alpha}{\alpha}, ... , \ketbra{\delta}{\delta}\}$ corresponding to a \emph{symmetric-informationally-complete}-POVM (or \emph{SIC-POVM}) on a qubit represented by the Bloch ball. In grey the convex hull of $\F_\mathrm{SIC}$ is given as regular tetrahedron. All quantum states $\Rho$ inside this tetrahedron are $\F_\mathrm{SIC}$-positive states and any CPTP map mapping this tetrahedron into itself are $\F_\mathrm{SIC}$-positive dynamical maps.}
    \label{fig:SIC-tetrahedron}
\end{figure}

In the non-Markovian case, the dynamical maps are no longer enough to infer the multi-time probabilities \cite{Milz2020}, and thus the possibility to simulate them via an invasive classical stochastic model. Consider the following simple example of a non-Markovian process that, despite being associated with an $\F_\sys$-positive dynamical map, cannot be simulated via the stochastic representation based on invasive measurements defined in \cref{sec:cim}. We have a two-level open quantum system, $\mathcal{H}_S $, interacting with a two-level environment, $\mathcal{H}_E $, so that the global evolution is fixed by the unitary map which acts between any considered time interval, i.e. from $t_0$ to $t_1$ as well as from $t_1$ to $t_2$,
\begin{equation}
    \Evl=e^{-\frac{i}{2}\left(\sigma_x \otimes \sigma_x + \sigma_y \otimes \sigma_y +2 \sigma_z \otimes \sigma_z\right)},
\end{equation}
and the initial environmental state $\tau_0 = \mathbbm{1}/2$. The resulting open-system CPTP map defined via \cref{eq:redmap} is easily seen to be a contraction of the Bloch ball, isotropic along the $x-y$ plan by an amount $\cos(1)\cos(2)$ while along the $z$-axis by an amount $\cos(1)^2$, so that the convex hull of the IC-POVM defined by the pure states $\left\{\ket{0}, \frac{1}{\sqrt{3}}\ket{0}+\sqrt{\frac{2}{3}}e^{i 2k \pi/3}\ket{1}\right\}_{k=1,2,3}$ is mapped into itself, i.e., the map is $\F_\sys$-positive. On the other hand, a direct evaluation of $P^{R_1,A1}(a_2; a_1|r_1)$ via \cref{eq:pr1a1}, shows that the quantity $P^{R_1,A1}(\ell_2;a_1| r_1)=\sum_{a_2}(M^{-1})_{\ell_2; a_2}P^{R_1,A1}(a_2; a_1|r_1)$ is not a probability distribution, since it takes on negative values, see \cref{tab:prob}; more details are given in \cref{app:cnsp}. Thus, because of \cref{pr:pr} there is no instantaneously-invasive informationally-complete stochastic process accounting for the same statistics; indeed, \cref{th:th} implies that this is due to the lack of $\F_\sys$-separability of the overall evolution.

\begin{table}
    \caption{\label{tab:prob}Values of $\sum_{a_1}P^{R_1,A1}(\ell_2;a_1| r_1)$ for different $\ell_2$ (rows) and $r_1$ (columns); indeed the negative values (in boldface) for fixed $r_1, \ell_2$ mean that at least one of the corresponding $P^{R_1,A1}(\ell_2;a_1| r_1)$ is negative, and it cannot be thus associated with a probability distribution. }
    \begin{ruledtabular}
    \begin{tabular}{llll}
    0.34 & 0.05 & 0.25 & {\bf -0.15}\\
    0.61 & 0.56 & 0.78 & 0.78\\
    0.18 & 0.08 & {\bf -0.15} & 0.28\\
    {\bf -0.13} & 0.31 & 0.12 & 0.09
    \end{tabular}
    \end{ruledtabular}
\end{table}

\section{Conclusion} 
In this paper, we have fully characterized a class of stochastic models that are invasive, but whose invasiveness can still be interpreted as having a classical origin. In particular, we have provided definite conditions that allow one, by looking at the statistics of the measurement outcomes, to tell whether such a classical model exists or not. Additionally, as our proof is constructive, one can use it to construct an explicit model, if it exists. We then identified a significant class of quantum processes that can be simulated by such a classical model. The analysis is focused on processes associated with sequential measurements of rank-one informationally-complete POVMs, deriving a sufficient condition to represent them via an invasive classical model that is connected with a definite property of the dynamics of the measured system. Furthermore, we have also shown, by means of an explicit example, that there are indeed quantum processes that cannot be simulated via the invasive models defined here.

This point deserves special attention since the fact that in quantum mechanics the measurement of a system alters its state is often understood as a major difference to classical physics or even \emph{the} peculiarity of quantum physics. However, our model and example suggests that the difference between classical and quantum physics is much more subtle than just the invasive character of measurements in the latter one. In this connection, it is also interesting to consider recent results that show how quantum mechanics can be modeled by a classical stochastic model, such as those presented in~\cite{Fritz2010,Hoffmann2018}. 

In this respect, a crucial constraint of our approach is the dimensionality of the classical invasive models taken into account. The very definition of the invasive measurement matrix in~\cref{ass1}, along with the completeness of the measurement expressed by~\cref{ass2} and the connection with the multi-time statistics in~\cref{ass3} essentially make the dimensionality of the classical model limited by the number of outcomes of the measured quantity. On the other hand, in~\cite{Fritz2010,Hoffmann2018} the internal state of the system that fixes the classical invasive model is not a-priori limited in dimensionality, which leads to the possibility to simulate all distributions that satisfy temporal ordering, thus including all quantum ones.

It will be an interesting task to generalise our results in various ways and deepen their connection to the existing literature. Recently, for instance, a generalization of the Kolmogorov consistency conditions has been brought forward with the idea to characterise quantum processes~\cite{Milz2020}. As the conditions presented here characterise an extended class of stochastic processes -- and hence are also direct generalisations of the same consistency conditions -- it will be interesting to study how these generalisations differ. In addition, we hope that the recent results on the dynamics of basis-dependent discord and coherence \cite{Smirne2018,strasberg_classical_2019,Milz2020a}, in relation to the non-classicality of time-correlations, can be seen as a limiting case of what we have investigated here. Indeed, the main difference consists in the type of measurement applied, as there orthogonal (but not complete) measurements were considered, while here we analyse complete (but not orthogonal) measurements. A further connection that is certainly worth investigating is with the theory of epsilon-transducers~\cite{Barnett2015,Budroni2022}, which has been used to calculate the memory needed to simulate a contextual experiment by a non-contextual one~\cite{Cabello2018}.

Finally, we note that the statistics considered in this paper refer to experiments and measurements that alter the state of a given system in a stochastic way, such that the result does not show the state before the measurement, but the state after it. Such measurements do not only appear in quantum mechanics, but may be important also in different scenarios \cite{Khrennikov2010}, where the fact that one does a measurement or experiment changes the outcomes. This is for instance a common problem in behavioural experiments, where the experiment does not show the natural behaviour of the subjects, but their behaviour under the experimental conditions.

\section*{Acknowledgment}
This work has been supported by the German Research Foundation (DFG) through FOR 5099;
A.S. acknowledges support from UniMi, via PSR-2 2021. D.E. acknowledges support from the Swiss National Science
Foundation (Grant No. P2SKP2 18406) and thanks the financial support from the Faculty of Science of Universidad de Los Andes, Bogotá, Colombia.

\bibliography{references.bib}

\cleardoublepage

\onecolumngrid

\appendix
\section{Extended statement and proof of Theorem \ref{pr:pr}}\label{app:the}
We begin by defining the two-time measurement-and-prepare statistics as the statistics that contains all (in principle) experimentally accessible probability distributions as laid out in the main text.
\begin{defin}
	A \textit{two-time measurement-and-prepare statistics from invasive measurements} is the collection of probability distributions: $(P(a_1), P(a_2), P^{A_1}(a_2;a_1), P^{R_1}(a_2|r_1), P^{R_1,A_1}(a_2;a_1|r_1), \linebreak[2] M_{a_1;\ell_1} = P(a_1|\ell_1), M_{a_2;\ell_2}=P(a_2|\ell_2))$,
	where the meaning of the different labels is explained in detail in the main text.
\end{defin}
Having clarified what quantities are considered, we now proceed by defining two models which may be used to explain the observed statistics. The first model is more in line with statistical descriptions, like the one used in Kolmogorov's theorem, while the second is directly defined in terms
of an open system, an environment and their interaction. 
\begin{defin}\label{def:sim_cont}
    We say that a two-time measurement-and-prepare statistics from invasive measurements can be simulated by a contextual model with instantaneously-invasive informationally-complete (IIIC) measurements  if and only if there is a probability distribution $P^{R_1, A_1}(a_2,\ell_2;a_1,\ell_1|r_1)$ that is consistent with the conditions 1-5 in the main text and from which the probabilities above can be obtained by \cref{eq:ecc} together with the KCCs over the corresponding $\ell_i$s. 
\end{defin}

\begin{defin}\label{def:sim_evol}
	We say that a \textit{two-time measurement-and-prepare statistics from invasive measurements} can be simulated by an open system stochastic evolution with IIIC measurements if and only if there are stochastic matrices $T_1((e_1,\ell_1);\ell_0)$ (with $\sum_{e_1,\ell_1} T_1((e_1,\ell_1);\ell_0)=1\forall \ell_0$) and $T_2(\ell_2;(e_1,\ell_1))$ (with $\sum_{\ell_2} T_2(\ell_2;(e_1,\ell_1))=1\forall e_1,\ell_1$), and a probability distribution $P(\ell_0)$ such that all the above probabilities can be calculated from the corresponding evolutions from $P(\ell_0)$ under the action of $T_1$ and $T_2$ by applying the measurements $M_{a_1,\ell_1}$ and $M_{a_2,\ell_2}$. 
	That is:
	\begin{align}
	   &P^{R_1,A_1}(a_2;a_1|r_1)=\sum_{\ell_2,\ell_1,e_1,\ell_0} M_{a_2; \ell_2 } T_2(\ell_2;(e_1,r_1))M_{a_1;\ell_1}T_1((e_1,\ell_1);\ell_0) P(\ell_0), \\ 
	   &P^{R_1}(a_2|r_1) =\sum_{\ell_2,\ell_1,e_1,\ell_0} M_{a_2; \ell_2 } T_2(\ell_2;(e_1,r_1))T_1((e_1,\ell_1);\ell_0) P(\ell_0),\\
	   &P(a_2) =\sum_{\ell_2,\ell_1,e_1,\ell_0} M_{a_2; \ell_2 } T_2(\ell_2;(e_1,\ell_1)) T_1((e_1,\ell_1);\ell_0) P(\ell_0),\\
	   &P^{A_1}(a_2;a_1) =\sum_{\ell_2,\ell_1,e_1,\ell_0} M_{a_2; \ell_2 } T_2(\ell_2;(e_1,a_1))M_{a_1;\ell_1}T_1((e_1,\ell_1);\ell_0) P(\ell_0),\\
	   &P(a_1)	=\sum_{\ell_1,e_1,\ell_0} M_{a_1;\ell_1}T_1((e_1,\ell_1);\ell_0) P(\ell_0). 
	\end{align}
\end{defin}
The above two models certainly feel very much related. Indeed one can test either of the two models by simply checking four conditions as stated in the following theorem which entails Theorem 1 of the main text
\begin{thm*}
	Let $S=(P(a_1), P(a_2), P^{A_1}(a_2;a_1), P^{R_1}(a_2|r_1),$ $ P^{R_1,A_1}(a_2;a_1|r_1), M_{a_1;\ell_1}=P(a_1|\ell_1), M_{a_2;\ell_2}=P(a_2|\ell_2))$ be a two-time measurement-and-prepare statistics from invasive measurements. Furthermore, let $M$ be invertible. Let $P(\ell_1):=\sum_{a_1} (M^{-1})_{\ell_1;a_1} P(a_1)$ and $P^{R_1}(\ell_2;\ell_1|r_1):=\sum_{a_1,a_2} (M^{-1})_{\ell_2;a_2}  (M^{-1})_{\ell_1;a_1} P^{R_1,A_1}(a_2;a_1|r_1)$. Then, the following three statements are equivalent.
	\begin{enumerate}
		\item The probability distributions associated with $S$ satisfy
		\begin{align}
			&P(\ell_1)	\geq 0 \text{ and } \sum_{\ell_1} P(\ell_1)=1,\label{cond:1}\\
			&P^{R_1}(\ell_2;\ell_1|r_1)\geq 0\text{ and } \sum_{\ell_2,\ell_1} P^{R_1}(\ell_2;\ell_1|r_1)=1\label{cond:2}\\
			&\sum_{\ell_2} P^{R_1}(\ell_2;\ell_1|r_1)=P(\ell_1) \; \forall r_1. \label{cond:5}\\
			&P^{R_1}(a_2|r_1)=\sum_{a_1} P^{R_1,A_1}(a_2;a_1|r_1) \label{cond:3}\\
			&P(a_2)=\sum_{a_1,r_1} (M^{-1})_{r_1;a_1} P^{R_1,A_1}(a_2;a_1|r_1), \label{cond:4}\\
			&P^{A_1}(a_2;a_1) = P^{R_1,A_1}(a_2;a_1|r_1=a_1)\label{cond:6}
		\end{align}
		\item $S$ can be simulated by an open system stochastic evolution with IIIC measurements.
		\item $S$ can be simulated by a contextual model with IIIC measurements.
	\end{enumerate}
    In the case that $P(\ell_1)$ and $P^{R_1}(\ell_2;\ell_1|r_1)$ are quasi probability distributions (and can have negative entries), the theorem holds up to $P^{R_1,A_1}(a_2, \ell_2; a_1, \ell_1|r_1)$ having negative entries and the corresponding evolutions can be quasi-stochastic.
\end{thm*}
We will prove this theorem by the steps $1. \Rightarrow 2.$, $2. \Rightarrow 3.$, and $3. \Rightarrow 1.$ For the first step, we will take a simple initial state $P_0(\ell_0):= \sum_{a_1}P(a_1)\delta_{a_1,\ell_0}$ (which is basically the same as the state $P(a_1)$)  and explicitly construct the matrices $T_1$ and $T_2$. We then proceed to show that these are indeed stochastic matrices and that all the conditions in statement 2. are satisfied. By construction, the probabilities one can generate from these maps satisfy conditions 1 to 5 in the main text, and directly from the conditions we get the right probabilities. The last step is outlined in the main text to motivate the conditions of statement 1. and here we will provide the details. 
\begin{proof}
	"$1. \Rightarrow 2.$":\\
	We define
	\begin{align}
		&P_0(\ell_0):= \sum_{a_1}P(a_1)\delta_{a_1,\ell_0},\\
		&T_1((e_1,\ell_1);\ell_0):=  \delta_{e_1,\ell_0} \delta_{\ell_1,\ell_0},\\
		&T_2(\ell_2;(e_1,r_1))
		:=
			\frac{\sum_{a_1,a_2} (M^{-1})_{\ell_2;a_2}  (M^{-1})_{e_1;a_1} P^{R_1,A_1}(a_2;a_1|r_1) }
                {\sum_{a_1}(M^{-1})_{e_1;a_1} P(a_1)}
		=\frac{P^{R_1}(\ell_2;\ell_1=e_1|r_1)}{P(\ell_1=e_1)}
	\end{align}
	with the convention that $0/0=1/n_{L2}$, with $n_{L2}$ the dimension of the space labelled by $\ell_2$. It directly follows from the definition that $T_1$ is a stochastic matrix and from condition~(\ref{cond:1}) that $P_0(\ell_0)$ is a probability distribution. $T_2$ is a stochastic map; it is positive, if both the nominator and denominator are positive, which is true by conditions~(\ref{cond:1} and \ref{cond:2}). Furthermore, by conditions~(\ref{cond:5}) and~(\ref{cond:1}) $T_2$ is a conditional probability and as such a stochastic map.	From the above definitions we get that 
	\begin{align*}
		&\sum_{\ell_2,\ell_1,e_1,\ell_0} M_{a_2; \ell_2 } T_2(\ell_2;(e_1,r_1)) M_{a_1;\ell_1}T_1((e_1,\ell_1);\ell_0) p(\ell_0)\\
		=&\sum_{\ell_2,\ell_1,e_1,\ell_0} M_{a_2; \ell_2 } 
		\frac{\sum_{a_1',a_2'} (M^{-1})_{\ell_2;a_2'}  (M^{-1})_{e_1;a_1'} P^{R_1,A_1}(a_2';a_1'|r_1) 
		}{\sum_{a_1'}(M^{-1})_{e_1;a_1'} P(a_1') 
		}
		M_{a_1;\ell_1}
		\sum_{a_1''} \delta_{e_1,\ell_0} \delta_{\ell_1,\ell_0}(M^{-1})_{\ell_1;a_1''} P(a_1'')\\
		=&\sum_{\ell_1}
		\frac{\sum_{a_1',a_2'}  \delta_{a_2,a_2'}  (M^{-1})_{\ell_1;a_1'} P^{R_1,A_1}(a_2';a_1'|r_1)  
		}{ \sum_{a_1'}(M^{-1})_{\ell_1;a_1'} P(a_1')
		}
		M_{a_1;\ell_1}
		\sum_{a_1''}  (M^{-1})_{\ell_1;a_1''} P(a_1'')\\
		=&\sum_{\ell_1} 
		\frac{\sum_{a_1'}  (M^{-1})_{\ell_1;a_1'} P^{R_1,A_1}(a_2;a_1'|r_1)  
		}{P(\ell_1)
		}
		M_{a_1;\ell_1}
		P(\ell_1)\\
		=&\sum_{\ell_1} 
		\sum_{a_1'}   (M^{-1})_{\ell_1;a_1'} P^{R_1,A_1}(a_2;a_1'|r_1) M_{a_1;\ell_1}
		=
		\sum_{a_1'}   \delta_{a_1,a_1'}  P^{R_1,A_1}(a_2;a_1'|r_1) 
		= P^{R_1,A_1}(a_2;a_1|r_1).
	\end{align*}
	This proves the first condition of~\cref{def:sim_evol}. The second condition then easily follows by using condition~(\ref{cond:6}),
	\begin{align*}
		&P^{A_1}(a_2;a_1)=P^{R_1,A_1}(a_2;a_1|r_1=a_1)
		=\sum_{\ell_2,\ell_1,e_1,\ell_0} M_{a_2; \ell_2 } T_2(\ell_2;(e_1,a_1))M_{a_1;\ell_1}T_1((e_1,\ell_1);\ell_0) p(\ell_0).
	\end{align*}
	For the third condition, we can insert the identity $\left[\sum_{a_1}(M^{-1})_{r_1;a_1} M_{a_1;\ell_1}\right]=\delta_{r_1,\ell_1}$, to get
	\begin{align*}
	   &\sum_{\ell_2,\ell_1,e_1,\ell_0} M_{a_2; \ell_2 } T_2(\ell_2;(e_1,\ell_1))T_1((e_1,\ell_1);\ell_0) p(\ell_0)
	   \\&=\sum_{\ell_2,\ell_1,e_1,\ell_0,r_1} M_{a_2; \ell_2 } T_2(\ell_2;(e_1,r_1)) \left[\sum_{a_1}(M^{-1})_{r_1;a_1} M_{a_1;\ell_1}\right] T_1((e_1,\ell_1);\ell_0) p(\ell_0)
	   \\&= \sum_{a_1,r_1} (M^{-1})_{r_1;a_1} P^{R_1,A_1}(a_2;a_1|r_1)
	   =P(a_2),
	\end{align*}
	where the last line follows from condition~(\ref{cond:4}). The condition $P(a_1)	=\sum_{\ell_1,e_1,\ell_0} M_{a_1;\ell_1}T_1((e_1,\ell_1);\ell_0) p(\ell_0)$ is trivially satisfied. 

    For the second last identity, we have that
    \begin{align*}
        &\sum_{\ell_2,\ell_1,e_1,\ell_0} M_{a_2; \ell_2 } T_2(\ell_2;(e_1,r_1)) T_1((e_1,\ell_1);\ell_0) p(\ell_0)\\
        =&\sum_{\ell_2,\ell_1,e_1,\ell_0} M_{a_2; \ell_2 } 
        \frac{\sum_{a_1',a_2'} (M^{-1})_{\ell_2;a_2'}  (M^{-1})_{e_1;a_1'} P^{R_1,A_1}(a_2';a_1'|r_1) 
        }{\sum_{a_1'}(M^{-1})_{e_1;a_1'} P(a_1') 
        }
        \sum_{a_1''} \delta_{e_1,\ell_0} \delta_{\ell_1,\ell_0}(M^{-1})_{\ell_1;a_1''} P(a_1'')\\
        =&\sum_{\ell_1}
        \frac{\sum_{a_1',a_2'}  \delta_{a_2,a_2'}  (M^{-1})_{\ell_1;a_1'} P^{R_1,A_1}(a_2';a_1'|r_1)  
        }{ \sum_{a_1'}(M^{-1})_{\ell_1;a_1'} P(a_1')
        }
        \sum_{a_1''}  (M^{-1})_{\ell_1;a_1''} P(a_1'')\\
        =&\sum_{\ell_1} 
        \sum_{a_1'}  (M^{-1})_{\ell_1;a_1'} P^{R_1,A_1}(a_2;a_1'|r_1)  
        =
        \sum_{a_1'}  P^{R_1,A_1}(a_2;a_1'|r_1) 
        = P^{R_1}(a_2|r_1), 
    \end{align*}
    where we have used that $M$ is a stochastic matrix and hence the columns of its inverse sum to one and~\cref{cond:3}. The last identity follows directly from the definitions. With this we have proven $"1. \Rightarrow 2."$ 

    The proof of $"2. \Rightarrow 3."$ is relatively straightforward. We define $P^{R_1,A_1}(a_2,\ell_2;a_1,\ell_1|r_1):=\sum_{e_1,\ell_0}  M_{a_2; \ell_2 } T_2(\ell_2;(e_1,r_1)) M_{a_1;\ell_1}T_1((e_1,\ell_1);\ell_0) p(\ell_0)$.  The statistics is then consistent with conditions 1, 2, 3 and 5 by construction, while condition 4 directly follows from $M$ being invertible. Finally, we get all of the probability distributions of~\cref{def:sim_cont} directly from the ones of~\cref{def:sim_evol}. In detail:
    \begin{align*}
 	      &P^{R_1,A_1}(a_2;a_1|r_1)
 	      =\sum_{\ell_2,\ell_1,e_1,\ell_0} M_{a_2; \ell_2 } T_2(\ell_2;(e_1,r_1))M_{a_1;\ell_1}T_1((e_1,\ell_1);\ell_0) P(\ell_0) = \sum_{\ell_2,\ell_1} P^{R_1,A_1}(a_2,\ell_2;a_1,\ell_1|r_1),
    \end{align*}
    meaning that the marginal over the unknown states $\ell_1$ and $\ell_2$ yields the measured probability distribution for   the case of doing all interventions.
    \begin{align*}
 	      &P^{A_1}(a_2;a_1)
 	      =\sum_{\ell_2,\ell_1,e_1,\ell_0} M_{a_2; \ell_2 } T_2(\ell_2;(e_1,a_1))M_{a_1;\ell_1}T_1((e_1,\ell_1);\ell_0) P(\ell_0) =\sum_{\ell_2,\ell_1} P^{R_1,A_1}(a_2,\ell_2;a_1,\ell_1|a_1),
    \end{align*}
    meaning that not re-preparing yields the same result as re-preparing in the measured state. 
    \begin{align*}
 	      &P^{R_1}(a_2|r_1)
 	      =\sum_{\ell_2,\ell_1,e_1,\ell_0} M_{a_2; \ell_2 } T_2(\ell_2;(e_1,r_1))T_1((e_1,\ell_1);\ell_0) P(\ell_0)
 	      \\
 	      =&\sum_{\ell_2,\ell_1,e_1,\ell_0} M_{a_2; \ell_2 } T_2(\ell_2;(e_1,r_1))
 	      \left(\sum_{a_1}M_{a_1; \ell_1 } \right)
 	      T_1((e_1,\ell_1);\ell_0) P(\ell_0)
 	      =\sum_{a_1,\ell_2,\ell_1} P^{R_1,A_1}(a_2,\ell_2;a_1,\ell_1|r_1),
    \end{align*}
    which means that we can take the marginal over $a_1$ in the usual way, as we delete the correlations with the environment by re-preparing the system.
    \begin{align*}
	   &P(a_2)
	     =\sum_{\ell_2,\ell_1,e_1,\ell_0} M_{a_2; \ell_2 } T_2(\ell_2;(e_1,\ell_1)) T_1((e_1,\ell_1);\ell_0) P(\ell_0)
	   \\=&
	   \sum_{\ell_2,\ell_1,e_1,\ell_0,r_1} M_{a_2; \ell_2 } T_2(\ell_2;(e_1,r_1))
	   \delta_{r_1,\ell_1}
	   T_1((e_1,\ell_1);\ell_0) P(\ell_0)
	   \\=&
	   \sum_{\ell_2,\ell_1,e_1,\ell_0,r_1} M_{a_2; \ell_2 } T_2(\ell_2;(e_1,r_1))
	   \left(\sum_{a_1}(M^{-1})_{r_1;a_1} M_{a_1;\ell_1}\right)
	   T_1((e_1,\ell_1);\ell_0) P(\ell_0)
	   \\=&\sum_{a_1,r_1,\ell_2,\ell_1} (M^{-1})_{r_1;a_1} 
	   \sum_{e_1,\ell_0} M_{a_2; \ell_2 } T_2(\ell_2;(e_1,r_1))
	   (M_{a_1;\ell_1}
	   T_1((e_1,\ell_1);\ell_0) P(\ell_0)\\
	   =&\sum_{a_1,r_1,\ell_2,\ell_1} (M^{-1})_{r_1;a_1}  P^{R_1,A_1}(a_2,\ell_2;a_1,\ell_1|r_1),
    \end{align*}
    where we do need to take into account the correlations with the environment at time $1$. Finally,
    \begin{align*}
 	      &P(a_1)
 	      =\sum_{\ell_1,e_1,\ell_0} M_{a_1;\ell_1}T_1((e_1,\ell_1);\ell_0) P(\ell_0)\\
 	      =&\sum_{\ell_1,e_1,\ell_0} \left(\sum_{a_2} M_{a_2; \ell_2 }\right) \left(\sum_{\ell_2}T_2(\ell_2;(e_1,r_1))\right) 
 	      M_{a_1;\ell_1}T_1((e_1,\ell_1);\ell_0) P(\ell_0)\\
 	      =& \sum_{a_2,\ell_2,\ell_1} P^{R_1,A_1}(a_2,\ell_2;a_1,\ell_1|r_1),
    \end{align*}
    which is just causality.

    We are left with showing $"3.\Rightarrow 1."$. First note that \cref{cond:2} implies \cref{cond:1} by virtue of condition 5, and \cref{cond:1} follows from the fact that
    $$
    P^{R_1}(\ell_2;\ell_1|r_1)=\sum_{a_1,a_2}P^{R_1,A_1}(a_2,\ell_2;a_1,\ell_1|r_1),
    $$
    with $P^{R_1,A_1}(a_2,\ell_2;a_1,\ell_1|r_1)$ a probability distribution by assumption.

    \cref{cond:5} is a direct consequence of causality (condition 5). \cref{cond:3} follows directly from condition 3 and the fact that the columns of $M^{-1}$ sum to one (being the inverse of a stochastic matrix), while  \cref{cond:6} follows from condition 4. The only condition left to check is \cref{cond:4}.
    \begin{align}
 	      &\sum_{a_1,r_1}(M^{-1})_{r_1;a_1}P^{R_1,A_1}(a_2; a_1|r_1)
 	      =\sum_{a_1,r_1,\ell_1}(M^{-1})_{r_1;a_1}P^{R_1,A_1}(a_2;  a_1,\ell_1|r_1)\nonumber
 	      \\&=\sum_{a_1,r_1,\ell_1}(M^{-1})_{r_1;a_1}M_{a_1;\ell_1}P^{R_1}(a_2; \ell_1|r_1)
 	      =\sum_{r_1,\ell_1}\delta_{r_1,\ell_1}P^{R_1}(a_2; \ell_1|r_1)\nonumber
 	      \\&=\sum_{\ell_1}P^{R_1}(a_2; \ell_1|r_1=\ell_1)
 	      =\sum_{\ell_1}P(a_2; \ell_1)
 	      =P(a_2),\nonumber
    \end{align}
    where the first equation follows from the KCC, the second equation from condition 3, the fifth from condition 5 and the last from the KCC.
 		 
    It follows directly from the proof, that quasi probability distributions $P(\ell_1)$ and $P^{R_1}(\ell_2;\ell_1|r_1)$, correspond to a quasi probability distribution $P^{R_1,A_1}(a_2, \ell_2; a_1, \ell_1|r_1)$ and quasi-stochastic evolutions.
\end{proof}

\section{Proof of Lemma \ref{lem:consistency}}\label{app:pol}
For convenience we reiterate the lemma here:
\begin{lem*}[\ref{lem:consistency}]
	A quantum stochastic process using IC-POVMs with probabilities as defined in the main text fulfils \cref{cond:reprep,cond:causal,eq:c1,eq:c2}. Furthermore, $P^{R_1}(\ell_2;\ell_1|r_1):=\sum_{a_1,a_2} (M^{-1})_{\ell_2;a_2}  (M^{-1})_{\ell_1;a_1} P^{R_1,A_1}(a_2;a_1|r_1)$ and $P(\ell_1)$ are quasi probability distribution (they sum to one, but are not necessarily positive).
\end{lem*}
\begin{proof}
    To start, note that
    \begin{align}
        K_{a_1} \Evl_0(\Rho_0) &= \sum_{\psi, \epsilon} \left(V_0 \vec{f}_0\right)_{(\psi, \epsilon)} K_{a_1} \ketbra{\psi}{\psi} K_{a_1}^\dagger \otimes \ketbra{\epsilon}{\epsilon} \nonumber
        \\
        &= \sum_{\psi, \epsilon} M_{a_1;\psi} \left(V_0 \vec{f}_0\right)_{(\psi, \epsilon)} \ketbra{a_1}{a_1} \otimes \ketbra{\epsilon}{\epsilon}
    \end{align}
    and hence
    \begin{align}
        \Pi_{r_1}\K_{a_1}\Evl_0(\Rho_0) = \sum_{\psi, \epsilon} M_{a_1;\psi} \left(V_0 \vec{f}_0\right)_{(\psi, \epsilon)} \ketbra{r_1}{r_1} \otimes \ketbra{\epsilon}{\epsilon}.
        \label{eq:Pi K V(Rho)}
    \end{align}
    Furthermore, with the quantum mechanical model given in the main text, we get the following expressions for the probabilities of interest:
    \begin{align}
        &P(a_1)  =  \Tr{\K_{a_1} \Evl_0(\Rho_0)}
        \label{eq:P(a_1)}
        \\
        &P(a_2)  =  \Tr{\K_{a_2} \Evl_1 \Evl_0(\Rho_0)}
        \label{eq:P(a_2)}
        \\&
        P^{A_1}(a_2;a_1)=\Tr{\K_{a_2} \Evl_1 \K_{a_1}\Evl_0(\Rho_0)}\label{eq:P(a_2;a1)}
        \\
        &P^{R_1}(a_2| r_1)  
        =\Tr{\K_{a_2} \Evl_1 \Pi_{r_1} \Evl_0(\Rho_0)}
        \label{eq:P(a_2;r_1)}
        \\
        &P^{R_1, A_1}(a_2; a_1|r_1) 
        \label{eq:P(a_2;r_1a_1)}
        =\Tr{\K_{a_2} \Evl_1 \Pi_{r_1} \K_{a_1}\Evl_0(\Rho_0)}.
    \end{align}
    For the lemma we have to proof that for a quantum process using an IC-POVM as quantum measurement the equations
    \begin{align}
        &P^{A_1}(a_2;a_1)
        =P^{R_1,A_1}(a_2;a_1|r_1=a_1)\label{eq:le1}\\
        &\sum_{\ell_2} P^{R_1}(\ell_2;\ell_1|r_1)=P(\ell_1) \; \forall r_1.\label{eq:le2}
        \\&
        P^{R_1}(a_2| r_1)= \sum_{a_1}P^{R_1,A_1}(a_2;a_1|r_1), \label{eq:le3}
        \\&
        P(a_2) =   \sum_{a_1,r_1}(M^{-1})_{r_1;a_1}P^{R_1,A_1}(a_2; a_1|r_1).\label{eq:le4}
    \end{align}
    hold. 

    For \cref{eq:le1}, we have that
    \begin{align*}
	   &P^{R_1,A_1}(a_2;a_1|r_1=a_1)=\Tr{\K_{a_2} \Evl_1 \Pi_{a_1} \K_{a_1}\Evl_0(\Rho_0)}
	   =\Tr{\K_{a_2} \Evl_1 \K_{a_1}\Evl_0(\Rho_0)} = P^{A_1}(a_2;a_1),
    \end{align*}
    since $\Pi_{a_1}\K_{a_1}= \K_{a_1}$, as the preparation simply discards any former state on the system and prepares the new one, but here both are identical. 

    For \cref{eq:le2}, we have that
    \begin{align*}
        &\sum_{\ell_2}P^{R_1}(\ell_2;\ell_1|r_1)
        =\sum_{\ell_2,a_1,a_2} (M^{-1})_{\ell_2;a_2}  (M^{-1})_{\ell_1;a_1} P^{R_1,A_1}(a_2;a_1|r_1)
        \\
        &=\sum_{\ell_2,a_1,a_2} (M^{-1})_{\ell_2;a_2}  (M^{-1})_{\ell_1;a_1} \Tr{\K_{a_2} \Evl_1 \Pi_{r_1} \K_{a_1} \Evl_0(\Rho_0)}
    \end{align*}
    To continue, we will introduce some notation to help the reader following our next steps. Let $\Evl_1 \Pi_{r_1} \K_{a_1} \Evl_0(\Rho_0) = \sum_{(\psi, \epsilon)} f'_{(\psi,\epsilon)} \ketbra{\psi}{\psi} \otimes \ketbra{\epsilon}{\epsilon}$. Accordingly we find
    \begin{align*}
        \K_{a_2}\Evl_1 \Pi_{r_1} \K_{a_1} \Evl_0(\Rho_0) = \sum_{(\psi,\epsilon)} M_{a_2;\psi} f'_{(\psi,\epsilon)} \ketbra{a_2}{a_2}\otimes\ketbra{\epsilon}{\epsilon}.
    \end{align*}
    Within the trace operation this gives
    \begin{align*}
        &\Tr{\K_{a_2}\Evl_1 \Pi_{r_1} \K_{a_1} \Evl_0(\Rho_0)} = \Tr{\sum_{(\psi,\epsilon)} M_{a_2;\psi} f'_{(\psi,\epsilon)} \ketbra{a_2}{a_2}\otimes\ketbra{\epsilon}{\epsilon}}
        \\
        &= \sum_{(\psi,\epsilon)} M_{a_2;\psi} f'_{(\psi,\epsilon)} \ubs{=1}{\Tr{\ketbra{a_2}{a_2}\otimes\ketbra{\epsilon}{\epsilon}}}
        = \sum_{(\psi,\epsilon)} M_{a_2;\psi} f'_{(\psi,\epsilon)} \ubs{=1}{\Tr{\ketbra{\psi}{\psi}\otimes\ketbra{\epsilon}{\epsilon}}}
        = \sum_{(\psi,\epsilon)} M_{a_2;\psi} \Tr{f'_{(\psi,\epsilon)} \ketbra{\psi}{\psi}\otimes\ketbra{\epsilon}{\epsilon}},
    \end{align*}
    where we have used that $\Tr{\ketbra{a_2}{a_2}\otimes\ketbra{\epsilon}{\epsilon}} = 1 = \Tr{\ketbra{\psi}{\psi}\otimes\ketbra{\epsilon}{\epsilon}}$ and hence we can exchange both expressions with each other. Thus, going back to the main calculation,
    \begin{align*}
        \sum_{\ell_2}P^{R_1}(\ell_2;\ell_1|r_1) = &\sum_{\ell_2,a_1,a_2, (\psi, \epsilon)} (M^{-1})_{\ell_2;a_2}  (M^{-1})_{\ell_1;a_1} M_{a_2;\psi}
        \Tr{f'_{(\psi,\epsilon)} \ketbra{\psi}{\psi}\otimes\ketbra{\epsilon}{\epsilon}}
        \\
        = &\sum_{\ell_2,a_1,(\psi, \epsilon)} \ubs{= \ (\id)_{\ell_2;\psi} \ = \ \delta_{\ell_2;\psi}}{\left(\sum_{a_2} (M^{-1})_{\ell_2;a_2} M_{a_2;\psi}\right)} (M^{-1})_{\ell_1;a_1}
        \Tr{f'_{(\psi,\epsilon)} \ketbra{\psi}{\psi}\otimes\ketbra{\epsilon}{\epsilon}}
        \\
        = &\sum_{a_1} (M^{-1})_{\ell_1;a_1} \Tr{\sum_{(\psi, \epsilon)} f'_{(\psi,\epsilon)} \ketbra{\psi}{\psi}\otimes\ketbra{\epsilon}{\epsilon}}
        \\
        = &\sum_{a_1} (M^{-1})_{\ell_1;a_1} \Tr{\Evl_1   \Pi_{r_1} \K_{a_1} \Evl_0(\Rho_0)}
        = \sum_{a_1} (M^{-1})_{\ell_1;a_1}\Tr{ \K_{a_1} \Evl_0(\Rho_0)}
        = \sum_{a_1} (M^{-1})_{\ell_1;a_1} P(a_1)
        = P(\ell_1)
    \end{align*}
    and to solve the sum over $a_1$ we have used the same procedure as for the sum over $a_2$ described above. 

    \cref{eq:le3} is straightforward:
    \begin{align*}
        &\sum_{a_1}P^{R_1,A_1}(a_2;a_1|r_1)
        =\sum_{a_1} \Tr{\K_{a_2} \Evl_1 \Pi_{r_1} \K_{a_1} \Evl_0(\Rho_0)}
        \\
        &= \Tr{\K_{a_2} \Evl_1  \Pi_{r_1} \sum_{a_1}\K_{a_1} \Evl_0(\Rho_0)}
        =\Tr{\K_{a_2} \Evl_1 \Pi_{r_1} \Evl_0(\Rho_0)} = P^{R_1}(a_2| r_1),
    \end{align*}
    where we have used $\sum_{a_1}\K_{a_1} = \id$. 

    Finally, to prove \cref{eq:le4}, note that
    \begin{align*}
        &\sum_{a_1,r_1}(M^{-1})_{r_1;a_1}P^{R_1,A_1}(a_2; a_1|r_1)=
        \sum_{a_1,r_1} (M^{-1})_{r_1;a_1}  \Tr{\K_{a_2} \Evl_1 \Pi_{r_1} \K_{a_1} \Evl_0(\Rho_0)}
        \\
        &=  \Tr{\K_{a_2} \Evl_1 \left(\sum_{a_1,r_1} (M^{-1})_{r_1;a_1}  \Pi_{r_1} \K_{a_1} \Evl_0(\Rho_0)\right)},
    \end{align*}
    which is equal to $P(a_2)=\Tr{\K_{a_2} \Evl_1 \Evl_0(\Rho_0)}$, if $\sum_{a_1,r_1} (M^{-1})_{r_1;a_1}  \Pi_{r_1} \K_{a_1} \Evl_0(\Rho_0)=\Evl_0(\Rho_0)$. This last equation can be seen by applying the frame decomposition:
    \begin{align*}
        &\sum_{r_1, a_1} (M^{-1})_{r_1;a_1} \Pi_{r_1} \K_{a_1} \Evl_0(\Rho_0)
        = \sum_{\psi, \epsilon} \sum_{r_1} \ubs{= \delta_{r_1,\psi}}{\left(\sum_{a_1} (M^{-1})_{r_1;a_1} M_{a_1;\psi}\right)} \left(V_0 \vec{f}_0\right)_{(\psi, \epsilon)} \ketbra{r_1}{r_1} \otimes \ketbra{\epsilon}{\epsilon}
        \\
        &= \sum_{\psi, \epsilon} \left(V_0 \vec{f}_0\right)_{(\psi, \epsilon)} \ketbra{\psi}{\psi} \otimes \ketbra{\epsilon}{\epsilon} = \Evl_0(\Rho_0).
    \end{align*}
    That $P^{R_1}(\ell_2;\ell_1|r_1):=\sum_{a_1,a_2} (M^{-1})_{\ell_2;a_2}  (M^{-1})_{\ell_1;a_1} P^{R_1,A_1}(a_2;a_1|r_1)$ and $P(\ell_1):=\sum_{a_1} (M^{-1})_{\ell_1;a_1} P(a_1)$ are quasi probability distributions, follows directly from the fact that $P(a_1)$ and $P^{R_1,A_1}(a_2;a_1|r_1)$ are probability distributions, while $M_{a_i;\ell_i}=\Tr{\E_{a_i} \ketbra{\ell_i}{\ell_i}} = \Tr{K_{a_i} \ketbra{\ell_i}{\ell_i} K_{a_i}^\dagger}$  are stochastic matrices due to the normalization condition $\sum_{a_i} \Tr{K_{a_i} \ketbra{\ell_i}{\ell_i} K_{a_i}^\dagger}=\Tr{\sum_{a_i} K_{a_i}^\dagger K_{a_i} \ketbra{\ell_i}{\ell_i} }=\Tr{\ketbra{\ell_i}{\ell_i}}=1$ (and hence their inverse are quasi-stochastic matrices).
\end{proof}

\section{Proof of Theorem \ref{th:th}}\label{app:th3}
\begin{thm*}[\ref{th:th}]
	A quantum process using an $\F_\sys$-based IC-POVM on $\Hils$ as measurement, a $\F_\sys$-separable initial state $\Rho_0$ and $\F_\sys$-separable unitaries $\Evl_{0}, \Evl_1 \in \mathcal{U}(\Hil)$ as initial and intermediate evolutions produce a proper stochastic probability distribution for all contexts.
\end{thm*}
To prove the theorem, we need to show that $P^{R_1}(\ell_2;\ell_1|r_1):=\sum_{a_1,a_2} (M^{-1})_{\ell_2;a_2}  (M^{-1})_{\ell_1;a_1} P^{R_1,A_1}(a_2;a_1|r_1)$ and $P(\ell_1):=\sum_{a_1} (M^{-1})_{\ell_1;a_1} P(a_1)$ are proper probability distributions with only positive entries. The theorem then follows directly from the lemma and \cref{pr:pr}. 

As explained in the main text, we can decompose a generic probability distribution in its frame decomposition
\begin{align}
    \Rho = \sum_{(\psi, \epsilon)} f_{(\psi, \epsilon)} \ketbra{\psi}{\psi} \otimes \ketbra{\epsilon}{\epsilon} = \sum_\psi f_\psi^\sys \ketbra{\psi}{\psi} \otimes \epsilonbm_\psi.
\end{align}
If $\Rho$ is $\F_\sys$-separable $f_\psi^\sys \geq 0$ and $\epsilonbm_\psi$ is a proper quantum state (in general $\epsilonbm_\psi$ is trace-one and hermitian for a minimal frame $\F_\sys$, but might be a indefinite or negative operator). Let us now define, for simplicity,
\begin{align}
    &\Rho:=\sum_\psi f_\psi^\sys \ketbra{\psi}{\psi} \otimes \epsilonbm_\psi:=\Evl_0(\Rho_0) 
    \\
    &\Rho'_{a_1,r_1} :=  \sum_\psi f_\psi'^\sys(a_1,r_1) \ketbra{\psi}{\psi} \otimes \epsilonbm'_\psi(a_1):=\frac{\Evl_1 \Pi_{r_1} \K_{a_1} \Evl_{0} \rho_0}{\Tr{\K_{a_1} \Evl_{0} \rho_0}}.
\end{align}
If $\Rho_0$, $\Evl_0$ and $\Evl_1$ are $\F_\sys$-separable the states $\Rho$ and $\Rho'_{a_1,r_1}$ are as well and hence, $f_\psi^\sys$, $f_\psi'^\sys(a_1,r_1) \geq 0$ and $\epsilonbm_\psi$, $\epsilonbm'_\psi(a_1)$ are proper quantum states.

With $M_{a_i;\ell_i}=\Tr{\E_{a_i} \ketbra{\ell_i}{\ell_i}} = \Tr{K_{a_i} \ketbra{\ell_i}{\ell_i} K_{a_i}^\dagger}$, we get that 
\begin{align*}
	&P(\ell_1):=\sum_{a_1} (M^{-1})_{\ell_1;a_1}P(a_1)
	=\sum_{a_1} (M^{-1})_{\ell_1;a_1} \Tr{\K_{a_1} \Evl_{0} \rho_0}
	\\&
	=\sum_{a_1} (M^{-1})_{\ell_1;a_1} \Tr{\K_{a_1} \sum_\psi f_\psi^\sys \ketbra{\psi}{\psi} \otimes \epsilonbm_\psi}
	\\&
	=\sum_{a_1} (M^{-1})_{\ell_1;a_1} \sum_\psi M_{a_1;\psi} f_\psi^\sys 
	=\sum_{\psi} \delta_{\ell_1,\psi} f_\psi^\sys 
	=f_{\ell_1}^\sys
	\geq 0
\end{align*}
and therefore
\begin{align*}
	&P^{R_1}(\ell_2;\ell_1|r_1)
	:=\sum_{a_1,a_2} (M^{-1})_{\ell_2;a_2}  (M^{-1})_{\ell_1;a_1} P^{R_1,A_1}(a_2;a_1|r_1)
	\\&
	:=\sum_{a_1,a_2} (M^{-1})_{\ell_2;a_2}  (M^{-1})_{\ell_1;a_1} \Tr{\K_{a_2}\Evl_1 \Pi_{r_1} \K_{a_1} \Evl_{0} \rho_0}
	\\&
	=\sum_{a_1,a_2} (M^{-1})_{\ell_2;a_2}  (M^{-1})_{\ell_1;a_1} 
	\Tr{\K_{a_2}\sum_\psi f_\psi'^\sys(a_1,r_1) \ketbra{\psi}{\psi} \otimes \epsilonbm'_\psi(a_1) \Tr{\K_{a_1} \Evl_{0} \rho_0}}
	\\&
	=\sum_{a_1,a_2} (M^{-1})_{\ell_2;a_2}  (M^{-1})_{\ell_1;a_1} \sum_\psi  M_{a_2;\psi} f_\psi'^\sys(a_1,r_1)  P(a_1)
	=\sum_{a_1} f_{\ell_{2}}'^\sys(a_1,r_1)  (M^{-1})_{\ell_1;a_1}P(a_1)
\geq 0,
\end{align*}
which ends the proof.\\

\section{Classically non-simulable process}\label{app:cnsp}
We report here more details on the process that cannot be simulated via classical invasive measurements discussed in the main text.

Both the system and the environment are two-level systems, $\mathcal{H}_S=\mathcal{H}_E=\mathbbm{C}^2$, and the global evolution is fixed by the unitary
\begin{equation}
    \Evl=e^{-\frac{i}{2}\left(\sigma_x \otimes \sigma_x + \sigma_y \otimes \sigma_y +2 \sigma_z \otimes \sigma_z\right)},\label{eq:evl}
\end{equation}
while the initial environmental state is $\tau_0 = \mathbbm{1}/2$. The CPTP dynamical maps that fix the open-system evolution in the absence of any intervention are thus given by -- compare with \cref{eq:redmap} --
\begin{eqnarray}\label{eq:mapsup}
    &&\Lambda(\Rho)
    =\partTr{\env}{\Evl (\Rho \otimes  \mathbbm{1}/2) \Evl^\dagger}
    =\frac{1}{2}\big(\Tr{\Rho}\mathbbm{1}+\cos(1)\cos(2) (\sigma_x \Tr{\sigma_x \Rho}
    +\sigma_y \Tr{\sigma_y \Rho})
    +\cos(1)^2 \sigma_z \Tr{\sigma_z \Rho}\big),
\end{eqnarray}
which corresponds to a contraction of the Bloch ball, isotropic along the $x-y$ plan by an amount $\cos(1)\cos(2)$ and by an amount $\cos(1)^2$ along the $z$-axis; here $\sigma_j$, $j=x,y,z$, are indeed the Pauli matrices and $\mathbbm{1}$ is the identity on $\mathbbm{C}^2$.

Considering the IC-POVM fixed by the pure states $\left\{\psi\right\}=\left\{\ket{0}, \frac{1}{\sqrt{3}}\ket{0}+\sqrt{\frac{2}{3}}e^{i 2k \pi/3}\ket{1}\right\}_{k=1,2,3}$, the map in Eq.(\ref{eq:mapsup}) is $\F_\sys$-positive with respect to the corresponding frame, i.e., it maps the regular tetrahedron corresponding to the convex hull of $\left\{\ket{\psi}\bra{\psi}\right\}$ into itself.  This can be verified by evaluating the FDCs of each of the four states given by $\Lambda(\ket{\psi}\bra{\psi})$. Using frame theory, for any state $\Rho$ the corresponding FDCs can be evaluated via the relation
\begin{equation}
    f_\psi(\Rho) = \bra{\psi}\mathbbm{S}^{-1}[\Rho]\ket{\psi},
\end{equation}
where $\mathbbm{S}^{-1}$ is the inverse of the map
\begin{equation}
    \mathbbm{S}(\Rho) = \sum_{\psi}\Tr{\ket{\psi}\bra{\psi} \Rho}\ket{\psi}\bra{\psi}.
\end{equation}
Using the Pauli matrices $\{\mathbbm{1}, \sigma_x, \sigma_y, \sigma_z\}$ as orthonormal basis for the space of Hermitian $2\times2$ matrices we know $\Rho = \frac{1}{2} \left(\Tr{\Rho} + \Tr{\sigma_x \Rho} + \Tr{\sigma_y \Rho} + \Tr{\sigma_z \Rho} \right)$ and we ca represent $\Rho$ by a vector $\vec{\rho}_\sigma = \frac{1}{2}(1, \Tr{ \vec{\sigma} \Rho} )^T$. One can show that in this orthonormal basis the super operator $\mathbbm{S}$ takes the form
\begin{equation}
    \mathbbm{S} = \begin{pmatrix} 2 & 0 & 0 & \\ 0 & \frac{2}{3} & 0 & 0 \\ 0 & 0 & \frac{2}{3} & 0 \\ 0 & 0 & 0 & \frac{2}{3} \end{pmatrix} \quad \und \quad \mathbbm{S}^{-1} = \begin{pmatrix} \frac{1}{2} & 0 & 0 & \\ 0 & \frac{3}{2} & 0 & 0 \\ 0 & 0 & \frac{3}{2} & 0 \\ 0 & 0 & 0 & \frac{3}{2} \end{pmatrix}.
\end{equation}
Accordingly we find
\begin{equation}
    \mathbbm{S}(\Rho) = \Tr{\Rho} \mathbbm{1}+\frac{1}{3}(\sigma_x \Tr{\sigma_x \Rho}
    +\sigma_y \Tr{\sigma_y \Rho}+\sigma_z \Tr{\sigma_z \Rho}\sigma_z)
\end{equation}
so that
\begin{equation}
    \mathbbm{S}^{-1}(\Rho) = \frac{1}{4}\Tr{\Rho}\mathbbm{1}+\frac{3}{4}(\sigma_x \Tr{\sigma_x \Rho}
    +\sigma_y \Tr{\sigma_y \Rho}+\sigma_z \Tr{\sigma_z \Rho}\sigma_z) 
\end{equation}
and the FDCs coefficients of the four evolved states $\left\{\Lambda(\ket{\psi}\bra{\psi})\right\}$ are reported in \cref{tab:FDCs}, from which one can see their positivity.
\begin{table}
    \caption{For the IC-POVM fixed by the pure states $\left\{\ket{\psi}\right\}=\left\{\ket{0}, \frac{1}{\sqrt{3}}\ket{0}+\sqrt{\frac{2}{3}}e^{i 2k \pi/3}\ket{1}\right\}_{k=1,2,3}$ this table lists the FDCs -- according to $\F = \left\{ \ketbra{\psi}{\psi} \right\}$ -- of the frame elements evolved by the CPTP map $\Lambda$, i.e. $f_\psi\left[ \ \Lambda(\ketbra{\psi'}{\psi'}) \ \right]$. The rows are indexed by $\psi$ and the columns by $\psi'$ using the abbreviations $a=\cos(1)\cos(2)\approx-0.22$ and $b=\cos(1)^2\approx0.29$. }
    \begin{ruledtabular}
    \begin{tabular}{llll}
        $\frac{1}{4}(1+3b)$&$\frac{1}{4}(1-b)$ & $\frac{1}{4}(1-b)$ &$\frac{1}{4}(1-b)$ \\
        $\frac{1}{4}(1-b)$&$\frac{1}{12}(3+8a+b)$ &$\frac{1}{12}(3-4a+b)$ &$\frac{1}{12}(3-4a+b)$ \\
        $\frac{1}{4}(1-b)$&$\frac{1}{12}(3-4a+b)$  &$\frac{1}{12}(3+8a+b)$ &$\frac{1}{12}(3-4a+b)$ \\
        $\frac{1}{4}(1-b)$&$\frac{1}{12}(3-4a+b)$ &$\frac{1}{12}(3-4a+b)$ &$\frac{1}{12}(3+8a+b)$
    \end{tabular}
    \end{ruledtabular}
    \label{tab:FDCs}
\end{table}

From the global unitary evolution in \cref{eq:evl}, we can indeed also evaluate all the multi-time probabilities associated with possible measurements and re-preparations at intermediate times. In particular, from \cref{eq:pr1a1} we get $P^{R_1,A1}(a_2; a_1|r_1)$; moreover, the matrix $M$ with elements $M_{a;\ell} = \Tr{K_a \ketbra{\ell}{\ell} K_a^\dagger}$ for the chose IC-POVM reads
\begin{eqnarray}
    M & = &\frac{1}{6}
                  \begin{pmatrix}
                             3 & 1 & 1 & 1\\
                             1 & 3 & 1 & 1\\
                             1 & 1 & 3 & 1\\
                             1 & 1 & 1 & 3
                 \end{pmatrix},
        \label{eq:ematrix}
\end{eqnarray}
and then one gets the values of $\sum_{a_1}P^{R_1,A1}(\ell_2;a_1| r_1)=\sum_{a_2}(M^{-1})_{\ell_2; a_2}P^{R_1,A1}(a_2; a_1|r_1)$ reported in \cref{tab:prob} in the main text.

\end{document}